\documentclass[sigconf,nonacm]{acmart}
\usepackage{enumitem}
\usepackage{balance}

\usepackage{amsmath,amssymb,amsfonts,amsthm}
\usepackage{algorithmic}
\usepackage{graphicx}
\usepackage{textcomp}
\usepackage{xcolor}
\usepackage{amsthm}
\usepackage{booktabs}
\usepackage{arydshln} 

\usepackage[table]{xcolor}
\usepackage{float}
\usepackage{makecell}
\usepackage{alltt}

\definecolor{randcol}{RGB}{214,183,165}
\definecolor{kindonlycol}{RGB}{190,143,112}
\definecolor{kindsemcol}{RGB}{123,125,103}
\definecolor{fullfpcol}{RGB}{165,162,132}
\definecolor{theorycol}{RGB}{205,203,188}

\usepackage{multirow}
\usepackage{float}
\usepackage{fontawesome}

\usepackage{colortbl}
\usepackage{threeparttable}
\usepackage{tikz}
\usetikzlibrary{positioning,arrows.meta,shapes.geometric,fit,calc}

\definecolor{MinColor}{rgb}{1, 0.9, 0.9}
\definecolor{MaxColor}{rgb}{0.6, 0, 0}
\definecolor{lightorange}{HTML}{FFA07A}
\definecolor{lightblue}{HTML}{ADD8E6}
\definecolor{fieldbg}{HTML}{FFF1C7}
\definecolor{revcolor}{rgb}{0, 0, 0}

\usepackage{tcolorbox}
\usepackage{subcaption}
\usepackage{listings}
\usepackage{xcolor}

\usepackage{hyperref}
\hypersetup{
    colorlinks=true,
    linkcolor=blue,
    filecolor=magenta,
    urlcolor=teal,
    }

\usepackage{url}

\usepackage{pifont}
\usepackage{bbding}
\newcommand{\cmark}{\textcolor{green!80!black}{\ding{51}}}
\newcommand{\xmark}{\textcolor{red}{\ding{55}}}
\newcommand{\pmark}{\textcolor{orange!90!black}{\raisebox{0.15ex}
{\scriptsize$\blacktriangle$}}}
\newcommand{\codefield}[1]{\begingroup\setlength{\fboxsep}{1pt}\colorbox{fieldbg}{\strut\texttt{#1}}\endgroup}

\newtheorem{theorem}{Theorem}

\newcommand{\heading}[1]{{\vspace{3pt}\noindent{\textbf{#1}}}}

\newenvironment{packeditemize}{
	\begin{list}{$\bullet$}{
			\setlength{\labelwidth}{4pt}
			\setlength{\itemsep}{0pt}
			\setlength{\leftmargin}{\labelwidth}
			\addtolength{\leftmargin}{\labelsep}
			\setlength{\parindent}{0pt}
			\setlength{\listparindent}{\parindent}
			\setlength{\parsep}{0pt}
			\setlength{\topsep}{1pt}}}{\end{list}}

\renewcommand\footnotetextcopyrightpermission[1]{} 
\begin{document}

\title{Five Attacks on x402 Agentic Payment Protocol}


\author{Zelin Li\textsuperscript{1},
Qin Wang\textsuperscript{2},
Zhipeng Wang\textsuperscript{3,\;{\footnotesize \faEnvelopeO}}}

\affiliation{%
\vspace{5pt}
  \institution{\textit{$^1$Ohio State University, US} $|$ \textit{$^2$CSIRO, Australia} $|$ \textit{$^3$The University of Manchester, UK}}
  \city{}
  \country{}
}

\begin{abstract}

The \textit{x402} protocol revives the HTTP \textit{402 Payment Required} status code to enable web-native micropayments across APIs, content, and agents. It combines synchronous HTTP authorization with asynchronous blockchain settlement and introduces a cross-layer attack surface absent from conventional web and on-chain payments.

In this paper, we formally analyze x402 and empirically show that it is vulnerable in both design and implementation. We present five concrete attacks that reveal weaknesses in authorization, binding, replay protection, and web-layer handling, showing that x402 is vulnerable across multiple stages of the payment workflow.

We validate these attacks through a reproducible testbed on local chains, Base Sepolia, and live endpoints and further audit three open-source SDKs and endpoints. Our results show that all five attacks are practical and can cause either \textit{unpaid} service or \textit{paid-but-denied} outcomes. We also propose practical mitigations.

\end{abstract}


\maketitle

\pagestyle{plain}
\thispagestyle{plain}

\section{Introduction}
\label{sec-intro}

The x402 protocol~\cite{x402-whitepaper} revives the long-reserved HTTP status code \texttt{402 Payment Required} as a web-native micro-payment handshake. Consider an autonomous AI agent that queries a weather API monetized through \textit{x402} (Fig.\ref{fig:x402model}). The interaction begins like an ordinary web request: the agent sends an HTTP \texttt{GET}, and the server replies with \texttt{402 Payment Required}, quoting a price such as 0.01\,USDC on Base. The agent signs a \textit{Payment Payload}. \ding{172} It then resends the request with this payload in the \texttt{X-PAYMENT} header. \ding{173} The server asks an off-chain facilitator to verify and settle the payment. \ding{174} The facilitator submits the transaction on chain. \ding{175} The server finally grants the requested resource. The result is a seconds-long paid API exchange, without API keys, account setup, or human intervention.

 \begin{figure}[t]
     \centering
    \includegraphics[width=0.91\linewidth]{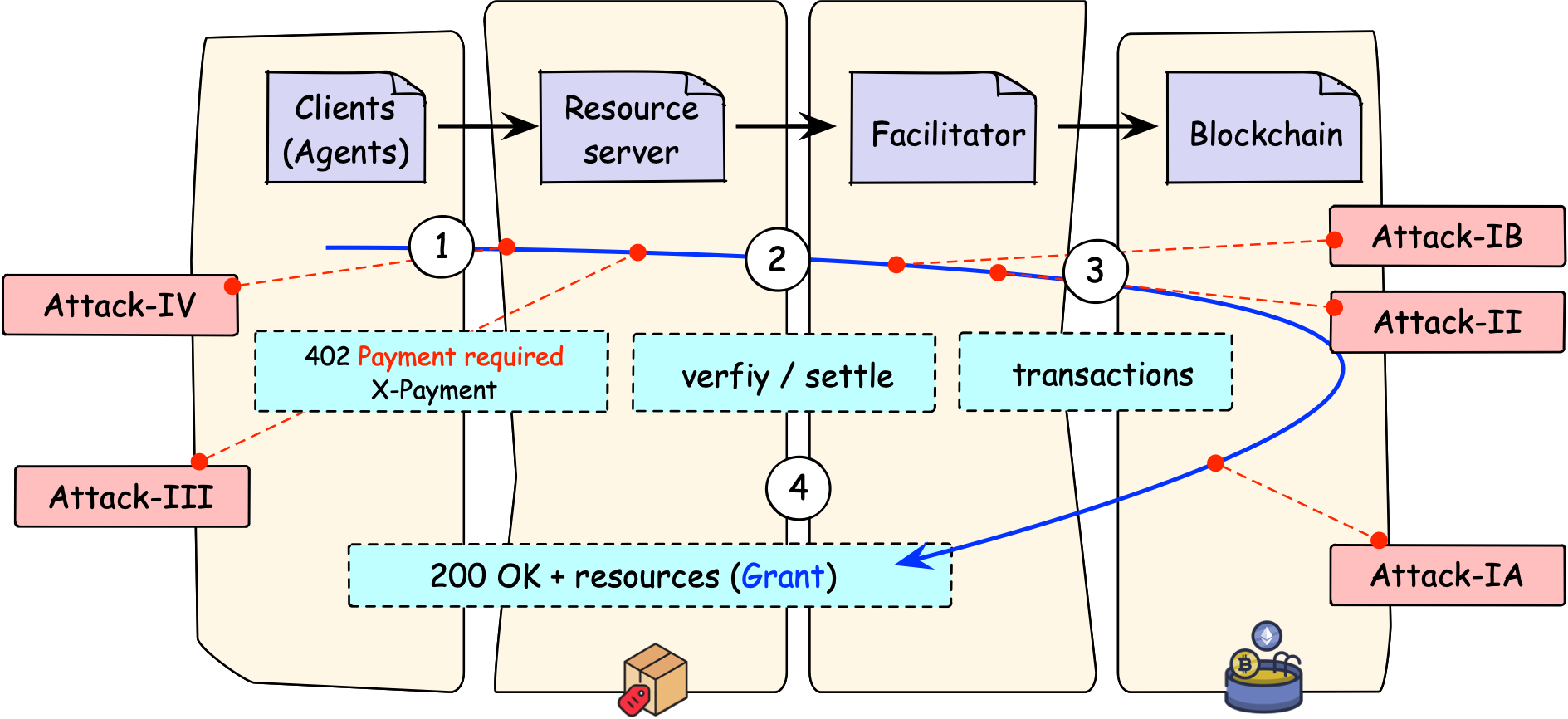}
    \vspace{-3pt}
   \caption{Simplified x402 workflow and our attack points. 
\ding{172} The client resends the request with \texttt{X-PAYMENT}.
\ding{173} The server asks the facilitator to verify and settle.
\ding{174} The facilitator submits the transaction on chain.
\ding{175} The server grants the resource after settlement. Our five attack points are in red.} 
\label{fig:x402-workflow}
     \label{fig:x402model}
     \vspace{-0.15in}
 \end{figure}

Since its public release by Coinbase in May 2025, the x402 ecosystem has grown rapidly: open-source SDKs exist in TypeScript, Python, Go, and .NET; a discovery layer (\emph{Bazaar}) catalogs over 13{,}000 registered resource servers; and the protocol is cited as the payment backbone for Agent-to-Agent (A2A) commerce~\cite{zhang2026sok,a2aproto}. 

\smallskip
Yet the apparent simplicity of x402 masks a \emph{cross-layer trust} gap. The HTTP layer operates synchronously: a request either succeeds or fails within milliseconds. The blockchain layer, by contrast, offers only \emph{probabilistic finality}: a transaction included in block~$b$ may be reverted by a reorganization of depth~$k$. x402 bridges these two layers through an off-chain facilitator whose correctness is neither enforced by the protocol nor verifiable by the client. 

This mismatch spans the phases and parties in the payment workflow, creating a large attack surface. Even before payment begins, the discovery layer may steer agents toward adversarial resource servers. Once the agent enters the payment flow, the \texttt{X-PAYMENT} header becomes a bearer-style payment capability that can be replayed or cached by intermediaries (\ding{172}). The server then relies on an off-chain facilitator for verification and settlement decisions, so a non-final or incorrect verification signal may still trigger access (\ding{173}). The settlement path may fail to bind the on-chain caller, facilitator, resource, and service decision into one atomic object (\ding{174}). Finally, after the server releases the resource, later settlement failure or cache reuse cannot be naturally rolled back (\ding{175}).

\vspace{3pt}
x402 should be treated as a cross-layer security protocol, not merely as an HTTP payment wrapper. However, to our knowledge, x402 has received rare prior formal security analysis.  Industry documentation focuses mainly on integration guides and happy-path examples, while the protocol specification omits several security-critical details. 
Existing academic work on payment channels~\cite{gudgeon2020sok,aumayr2022sleepy}, probabilistic micropayments~\cite{almashaqbeh2020microcash}, service execution \cite{li2026a402}, and API security~\cite{jayasuriya2024understanding} does not fully address x402's boundaries, where synchronous web authorization is tightly coupled with asynchronous blockchain settlement.

\begin{table*}[t]
\centering
\caption{Summary of our attacks on x402 payment protocols. RGP = revert-grant probability; DGR = duplicate-grant rate.}
\label{tab:attack-summary}
\scriptsize
\renewcommand{\arraystretch}{1.02}
\setlength{\tabcolsep}{4pt}
\resizebox{\textwidth}{!}{%
\begin{tabular}{@{}llll@{}}
\toprule
\multicolumn{1}{c}{\textbf{Attack}} & 
\multicolumn{1}{c}{\textbf{Mechanism}} & 
\multicolumn{1}{c}{\textbf{Impact}} & 
\multicolumn{1}{c}{\textbf{Experimental Results}} \\
\midrule
\textbf{I-A: Revert-grant}
& Grant before finality
& Unpaid service after settlement failure
& RGP$_0$ 5.18\%; Byzantine: 100\%
\\

\textbf{I-B: Settlement preemption}
& Caller-unbound Permit2 \texttt{settle()}
& Payment without service
& Base Sepolia: preemption yields HTTP~402

\\

\textbf{II: Replay / Idempotency}
& No idempotency/binding
& One payment reused for many grants
& DGR$\,{=}\,n$; live: 248 grants/payment
\\

\textbf{III: Header / Proxy confusion}
& Header ambiguity; caching
& Verification confusion and paid-content leak
& nginx leak: 100\%; live CDN leak
\\

\textbf{IV: Server-selection attacks}
& Metadata gaming + Sybils
& Adversarial capture of agent traffic
& One crafted server: 71.8\%; five Sybils: 60.2\%
\\
\bottomrule
\end{tabular}

}
\end{table*}

\vspace{3pt}
In this paper, we investigate the following question: 


\vspace{3pt}
\textit{
What are the concrete security risks in x402's current design and implementations, and how can they be mitigated without sacrificing its web-native usability?}
\vspace{3pt}

We answer this question through four attack classes, a targeted implementation audit, and live testnet validation (Table \ref{tab:attack-summary}):

\begin{packeditemize}
\item \textbf{Attack~I: settlement-path inconsistencies.}
We study failures between HTTP access control and on-chain settlement. Attack~I-A shows that optimistic execution may grant resources before payment is final. Attack~I-B shows that caller-unbound settlement lets an observer consume the payment authorization before the legitimate facilitator.

\item \textbf{Attack~II: replay and idempotency.}
We analyze how a reusable \texttt{X-PAYMENT} payload can produce multiple HTTP-layer grants when the server does not atomically record a payment identity before releasing the resource.

\item \textbf{Attack~III: web-layer handling.}
We examine how ordinary HTTP infrastructure interacts with payment-gated content. In particular, we validate cache leakage as an end-to-end failure and discuss header ambiguity as a deployment-dependent parser risk.

\item \textbf{Attack~IV: server selection.}
We study how agent-side discovery can be manipulated before the payment protocol even begins. We empirically validate metadata manipulation and Sybil flooding in Bazaar-style server selection, showing that adversaries can bias agents toward malicious paid endpoints.
\end{packeditemize}

These attacks cover the main stages of the x402 workflow: discovery, payment presentation, facilitator verification, on-chain settlement, and final resource delivery. We evaluate them across controlled, testnet, and implementation-level studies.

\vspace{3pt}
Specifically, we summarize our contributions as follows.

\begin{packeditemize}
\item We present a \textit{formalized security abstraction for x402}. We model the main entities, trust boundaries, workflow stages, and settlement assumptions, and abstract four security properties from this model: authorization soundness, payment--service correspondence, replay resistance, and facilitator $k$-atomicity. These properties provide a structured basis for analyzing where x402 holds, where it fails, and how different execution policies change the resulting risk (Theorems~\ref{thm:auth-sound}--\ref{thm:atomicity}, \S\ref{sec:formal-model}--\ref{sec:threat-model}).
 
\item We characterize \textit{five attacks across four classes} with measurable metrics. As shown in Table \ref{tab:attack-summary}, the attacks include settlement-path inconsistencies including revert-grant under optimistic execution (Attack~I-A) and unauthorized settlement preemption (Attack~I-B), replay/idempotency attacks (Attack~II), header/proxy confusion (Attack~III), and agent server-selection manipulation (Attack~IV), each specified with violation events, adversary capability sets, and quantitative indicators (\S\ref{sec-method}).
 
\item We implement a \textit{reproducible testbed} and run over 25{,}000 payment requests across 48~configurations on Hardhat/Anvil and Base Sepolia, supplemented by validation on four production x402 endpoints. The experiments confirm all attack classes, including RGP$_0$ up to 5.18\% with honest facilitators, DGR$\,{=}\,n$ without idempotency, on-chain settlement preemption, production CDN cache leakage, and agent selection bias up to 71.8\% under metadata manipulation and 60.2\% under 5-Sybil flooding. We report 95\% Wilson score confidence intervals for rate metrics (\S\ref{sec-evalu}).

\item We conduct a \textit{systematic multi-implementation audit} of three open-source x402~SDKs and four live endpoints, identifying 11~vulnerabilities across five classes. These include grant-before-settle behavior in a third-party Python SDK, missing resource-identifier binding, fire-and-forget settlement, and absent \texttt{Cache-Control} headers, with severity ratings (\S\ref{sec:audit}).

\item We propose \textit{actionable mitigations}, such as canonical encoding, two-phase settlement, mandatory idempotency with resource binding, and caching hygiene, and report findings via coordinated vulnerability disclosure (\S\ref{sec:discussion}).
\end{packeditemize}

\begin{center}
\fbox{%
\begin{minipage}{0.95\linewidth}
\textit{Responsible disclosure.} 
We privately reported our findings to Coinbase via HackerOne (\#3679163, \#3679179, \#3679220), including issues on x402 permit settlement, Bazaar discovery registration and resource-identifier binding.
\end{minipage}
}
\end{center}

\section{x402 Models and Protocols}
\label{sec-model}

This section presents the protocol flow, assumptions, security properties, and our theoretical results.

\subsection{x402 Overview}

The \textit{x402} protocol reinterprets the HTTP status \texttt{402 Payment Required} as a web-native payment handshake. A typical exchange proceeds in three phases:

\begin{enumerate}[leftmargin = *]
  \item \textbf{Request and quote.} The client sends an ordinary HTTP request, and the server replies with \texttt{402 Payment Required} and a \textit{Payment Requirements} object that quotes the payment terms.
  \item \textbf{Payment presentation.} The client constructs a signed \textit{Payment Payload} and resends the request via the \texttt{X-PAYMENT} header.
  \item \textbf{Verification, settlement, and grant.} The server asks the facilitator to verify and settle the payment, the facilitator submits the transaction on chain, and the server grants the protected resource once settlement succeeds.
\end{enumerate}

\subsection{System Model}\label{sec:formal-model}

We first define the entities, messages, and protocols.

\heading{Notation and entities.}
We model an \textbf{x402 system} as a tuple
$\mathcal{S} = (\mathcal{C}, \mathcal{R}, \mathcal{F}, \mathcal{B}, \mathcal{N}, \mathcal{T}, \lambda)$.
Here, $\mathcal{C} = \{C_i\}$ denotes the set of \emph{clients};
$\mathcal{R} = \{R_j\}$ the set of \emph{resource servers} that host protected resources and issue HTTP 402 responses;
$\mathcal{F} = \{F_\ell\}$ the set of optional \emph{facilitators} that verify payments and/or coordinate on-chain settlement;
$\mathcal{B}$ the \emph{blockchain ledger} used for settlement, modeled as an append-only sequence of blocks with probabilistic finality;
$\mathcal{N}$ the \emph{network and communication infrastructure}, including HTTP/TLS channels, proxies, caches, and blockchain mempools;
$\mathcal{T}$ the global logical time domain; and $\lambda$ the security parameter.

We write $(x \xrightarrow{\;m\;} y, t)$ to denote that at time $t$ entity $x$ sends message $m$ to entity $y$ over $\mathcal{N}$.

\heading{x402 messages.}
We define the primary structured messages exchanged in an x402 protocol execution:
\[
\begin{aligned}
  \mathit{PR} &:= \langle \mathit{resource\_id}, \mathit{amount}, \mathit{token},
  \mathit{chain\_id}, \mathit{receiver}, \mathit{expiry}, \mathit{meta}\rangle \\
  \mathit{PP} &:= \langle \mathit{payment\_id}, \mathit{payer\_addr}, \mathit{amount}, \mathit{chain\_id}, \mathit{nonce}, \mathit{ts}, \sigma \rangle
\end{aligned}
\]
Here, $\mathit{PR}$ denotes \textit{PaymentRequirements} and $\mathit{PP}$ denotes \textit{PaymentPayload}. The symbol $\sigma$ denotes a digital signature produced by the payer over a canonical serialization of the payload fields. We let $\mathit{pay\_id}=\mathit{PP.payment\_id}$ denote a unique identifier derived by the client or facilitator.

\heading{Blockchain model.}
The blockchain $\mathcal{B}$ exposes the following primitives:
$\mathsf{Broadcast}(tx)$, which submits a transaction to the mempool;
$\mathsf{Include}(tx, b)$, which records that $tx$ appears in block $b$ at height $h$;
$\mathsf{Confirmations}(tx, t)$, which returns the number of confirmations observed for $tx$ at time $t$;
$\mathsf{Final}(tx, k)$, which checks whether $tx$ has reached the $k$-confirmation finality threshold; and
$\mathsf{Reorg}(b)$, which captures a chain reorganization that may remove previously included transactions from the canonical chain.
We model inclusion and finality with probabilistic latency: for a broadcast at time $t_0$, inclusion time is a random variable $T_{\text{inc}}$ and finality occurs at $t_0 + T_{\text{fin}}$. We denote the block interval as $T_b$.

\begin{definition}[Chain Finality Bound]
\label{def:chain-finality}
The blockchain $\mathcal{B}$ satisfies the $k$-deep finality property if $\Pr[\mathsf{Reorg}(\text{depth} \ge k)] \le \epsilon_{\mathrm{chain}}(k)$,
where $\epsilon_{\mathrm{chain}}(k)$ is a monotonically decreasing function of $k$. For chains with exponential finality decay, $\epsilon_{\mathrm{chain}}(k) = e^{-\alpha k}$ for a chain-specific constant $\alpha > 0$.
\end{definition}

\heading{x402 protocol.}
An x402 \emph{execution trace} $\mathcal{E}_{x402}$ represents the ordered interactions among the client ($C$), resource server ($R$), facilitator ($F$), and blockchain ($\mathcal{B}$). Specifically, $\mathcal{E}_{x402}$ is a finite sequence of events: $\mathcal{E} = (e_1, e_2, \dots, e_n)$, $e_i \in \mathcal{M} \cup \mathcal{O}$, where $\mathcal{M}$ are off-chain messages (HTTP requests/responses, facilitator API calls) and $\mathcal{O}$ are on-chain observations (transactions broadcast, included, reverted).

Following the official x402 protocol specification~\cite{x402-whitepaper}, we decompose the canonical trace into three subprotocols:
$$\Pi_{x402} \;=\; \Pi_{\mathrm{rq}} \circ \Pi_{\mathrm{pp}} \circ \Pi_{\mathrm{vs}},$$
where $\Pi_{\mathrm{rq}}$ is the request-and-quote subprotocol, $\Pi_{\mathrm{pp}}$ is the payment-presentation subprotocol, and $\Pi_{\mathrm{vs}}$ is the verification-and-settlement subprotocol that culminates in a resource grant.

\vspace{3pt}
\noindent $\Pi_{\mathrm{rq}}$: \codefield{\textit{Request-and-quote subprotocol.}}
The client issues an ordinary web request and receives the quoted payment terms.
\begin{packeditemize}
\item $C \xrightarrow{\texttt{GET /api}} R$: the client requests a protected resource.
\item $R \xrightarrow{\texttt{402 Payment Required}} C$: the server returns one or more \texttt{payment requirements} that quote the amount, token, receiver, and expiry.
\item $C$: the client parses $\mathit{PR}$ and selects an admissible payment option.
\end{packeditemize}

\vspace{3pt}
\noindent $\Pi_{\mathrm{pp}}$: \codefield{\textit{Payment-presentation.}}
The client signs a Payment Payload $\mathit{PP}$ and resends the request in \texttt{X-PAYMENT} header.
\begin{packeditemize}
\item $C$: the client selects a payment method and constructs $\mathit{PP}$.
\item $C \xrightarrow{\texttt{X-PAYMENT:}~\mathit{PP}} R$: the client resends the request with $\mathit{PP}$.
\end{packeditemize}

\vspace{3pt}
\noindent $\Pi_{\mathrm{vs}}$: \codefield{\textit{Verification, settlement, and grant.}}
The server asks the facilitator to verify and settle, the transaction is submitted on chain, and the resource is granted on success.
\begin{packeditemize}
\item $R \xrightarrow{\texttt{/verify}~\langle \mathit{PP},\mathit{PR} \rangle} F$: the server asks the facilitator to verify the presented payment, or equivalently verifies locally.
\item $F \xrightarrow{\texttt{Verifi. Response}} R$: the facilitator returns its validation verdict.
\item $R \xrightarrow{\texttt{/settle}~\langle \mathit{PP},\mathit{PR} \rangle} F$: if verification succeeds, the server asks the facilitator to settle.
\item $F \xrightarrow{\texttt{tx}_{pp}} \mathcal{B}$: the facilitator submits the settlement transaction $\texttt{tx}_{pp}$ on-chain.
\item $\mathcal{B} \xrightarrow{\texttt{confirm}(k)} F$: the chain confirms the transaction to depth $k$.
\item $F \xrightarrow{\texttt{settled}} R$: the facilitator reports settlement status to the server.
\item $R \xrightarrow{\texttt{200 OK, X-PAYMENT-RESPONSE}} C$: the server grants the requested resource after the settlement and emits the terminal application-layer response.
\end{packeditemize}


\subsection{Threat Model}
\label{sec:threat-model}

We then specify the adversarial capabilities and our assumptions.

We consider a probabilistic polynomial-time (PPT) adversary
$\mathcal{A}$ that interacts with the x402 system
$\mathcal{S} = (\mathcal{C}, \mathcal{R}, \mathcal{F}, \mathcal{B}, \mathcal{N}, \mathcal{T}, \lambda)$
as defined in \S\ref{sec:formal-model}.
The adversary may possess different classes of capabilities,
denoted $\mathcal{A}_X$ for domain $X \in
\{\mathcal{N}, \mathcal{C}, \mathcal{F}, \mathcal{B}, \mathcal{R}\}$.

\emph{Adversarial capabilities} are defined as follows:
\begin{packeditemize}

  \item \textit{Network adversary} $\mathcal{A}_{\mathcal{N}}$:
  Can delay, drop, reorder, or replay HTTP messages, consistent with the network layer of $\mathcal{N}$.
  TLS channels guarantee message integrity and confidentiality,
  but $\mathcal{A}_{\mathcal{N}}$ can control timing and scheduling.

  \item \textit{Client adversary} $\mathcal{A}_{\mathcal{C}}$:
  May corrupt an arbitrary subset
  $\mathcal{C}_{\mathcal{A}} \subseteq \mathcal{C}$,
  enabling generation of malformed or replayed Payment Payloads,
  or deviation from honest protocol behavior.

  \item \textit{Facilitator adversary} $\mathcal{A}_{\mathcal{F}}$:
  May corrupt facilitators
  $\mathcal{F}_{\mathcal{A}} \subseteq \mathcal{F}$,
  causing Byzantine behavior such as returning false verifications or prematurely asserting settlement.

  \item \textit{Blockchain adversary} $\mathcal{A}_{\mathcal{B}}$:
  Can observe and manipulate the public mempool;
  may inject, reorder, or replace transactions
  within consensus constraints.
  Cannot violate the chain finality bound (Definition~\ref{def:chain-finality}):
  $\Pr[\text{reorg depth} > k] \le \epsilon_{\text{chain}}(k)$.

  \item \textit{Server adversary} $\mathcal{A}_{\mathcal{R}}$:
  Resource servers are assumed \emph{honest-but-curious}:
  they follow the protocol but may record observable data.
  Fully malicious servers are discussed in \S\ref{sec:discussion}.
\end{packeditemize}

\heading{Structural assumptions.} We make two structural assumptions:

\begin{packeditemize}
    \item \textit{Blockchain independence.}\label{asm:blockchain-indep} Chain reorganizations depend only on the honest/adversarial mining or staking power and block-propagation delays. Conditioned on the broadcast transactions, whether $\mathsf{Reorg}(\mathrm{depth} \ge k)$ occurs is independent of the facilitator's internal state and HTTP-layer events.

\item \textit{Facilitator determinism.} \label{asm:facilitator-det}
An honest facilitator $F$'s report $F.\mathsf{reportFinal}(tx, t)$ is
a deterministic function of its local view of $\mathcal{B}$ at time
$t$. A Byzantine facilitator is modeled as an arbitrary PPT machine.
\end{packeditemize}

\heading{Cryptographic assumptions.}
All cryptographic primitives (e.g., digital signatures, hash functions) used in the system are secure. The adversary can only break them with negligible probability in $\lambda$.

\heading{Composed adversaries.}
Adversaries may combine capabilities:
$\mathcal{A} = \bigcup_{X \in \mathbb{A}} \mathcal{A}_X$ where
$\mathbb{A} \subseteq \{\mathcal{N}, \mathcal{C}, \mathcal{F}, \mathcal{B}, \mathcal{R}\}$, e.g.,
$\mathcal{A}_{\mathcal{N}\mathcal{C}}\text{=}\mathcal{A}_{\mathcal{N}} \cup \mathcal{A}_{\mathcal{C}}$
models a network adversary who also controls a subset of clients.

\heading{Probability space.} 
Suppose a security parameter $\lambda$, we define the probability space
$(\Omega, \mathcal{F}_\sigma, \Pr_\lambda)$ where $\Omega$ is the set
of all execution traces $\mathcal{E}$ generated by a PPT adversary
$\mathcal{A}$ interacting with system $\mathcal{S}$. The
$\sigma$-algebra $\mathcal{F}_\sigma$ is generated by cylinder sets
over finite event prefixes. All probabilities are taken over the joint
randomness of $\mathcal{A}$, the blockchain's block-production process,
and network scheduling.

\subsection{Security Properties}
\label{sec:security-props}

We then define the security goals. A secure x402 system satisfies the following properties, each defined relative to an adversary $\mathcal{A}$ generating executions $\mathcal{E} \leftarrow \mathcal{A}(\mathcal{S})$.

\begin{definition}[Authorization Soundness]\label{def:auth-soundness}
An x402 system achieves \emph{authorization soundness} if, for every PPT adversary $\mathcal{A}$,
the probability that a server $R$ grants access to a client $C$
for resource $res$ at time $t$ without eventual final on-chain settlement
is negligible. Let
\[
E_{\text{auth}} =
\{\exists\, t,\, C,R,res,\mathit{PP} :
R.\text{grant}(C,res,t)
\wedge
\neg\,\mathsf{Final}(tx_{pp},k)\}.
\]
Authorization soundness holds if
$\Pr_{\mathcal{E} \leftarrow \mathcal{A}}[E_{\text{auth}}]
\le \mathsf{negl}(\lambda)$.
\end{definition}

\begin{definition}[Payment--Service Correspondence]\label{def:no-double}
Let $\mathsf{GrantCount}$ $(C,res,\mathcal{E})$ denote the number of
distinct grants for $(C,res)$, and
$\mathsf{Settlements}(\mathit{PP},\mathcal{E})$ the number of
on-chain settlements for the same logical payment identifier.
The consistency failure event is
\begin{align*}
E_{\text{dc}} =
\{\exists\, C,res,\mathit{PP} :
\mathsf{GrantCount}(C,res,\mathcal{E}) \neq
\mathsf{Settlements}(\mathit{PP},\mathcal{E})\}.
\end{align*}

This captures two failure modes:

\begin{packeditemize}
\item grants exceeding settlements (resource server delivers unpaid content, Attack~II);
\item settlement without service (payer charged but denied access, Attack~I-B).
\end{packeditemize}
Payment--service correspondence holds if $\Pr_{\mathcal{E} \leftarrow \mathcal{A}}[E_{\text{dc}}]
\le \mathsf{negl}(\lambda)$.

\end{definition}

\begin{definition}[Replay Resistance]\label{def:replay}
Let $\mathit{PP}$ be a valid PaymentPayload first used at time $t$.
The replay event is
\[
E_{\text{replay}} =
\{\exists\, t' > t :
R.\text{grant}(C,res,t') \wedge
\text{same } \mathit{PP} \text{ reused}\}.
\]
The system is \emph{replay-resistant} if
$\Pr_{\mathcal{E} \leftarrow \mathcal{A}}[E_{\text{replay}}]
\le \mathsf{negl}(\lambda)$.
\end{definition}

\begin{definition}[Facilitator $k$-Atomicity]\label{def:atomicity}
Let $E_{\text{atom}}$ be the event that a facilitator
reports a transaction $tx$ as final before $k$ confirmations:
\[
E_{\text{atom}} =
\{\exists\, tx,\,t :
\mathsf{Confirmations}(tx,t) < k
\wedge
F.\text{reportFinal}(tx,t)\}.
\]
A facilitator is $k$-atomic if
$\Pr_{\mathcal{E} \leftarrow \mathcal{A}}[E_{\text{atom}}]
\le \mathsf{negl}(\lambda)$.
\end{definition}

\subsection{Security Theorems}
\label{sec:theorems}

We now state the main consequences of the model above. The theorems formalize when these security properties hold or fail as a function of execution policy, facilitator honesty, and chain finality. 
We only show key statements; full proofs are placed at Appendix~\ref{app:proofs}.

\begin{theorem}[Authorization Soundness---Conservative Execution]
\label{thm:auth-sound}
Let $\mathcal{S}$ be an x402 system under \emph{conservative execution},
where $R$ grants access only upon receiving $F$'s report that
$\mathsf{Confirmations}(tx_{pp}, t) \ge k$.
Assume (i)~$F$ is honest (i.e., $\mathcal{F} \notin \mathbb{A}$),
(ii)~Assumption~\ref{asm:blockchain-indep} holds, and
(iii)~the blockchain satisfies the $k$-deep finality property
(Definition~\ref{def:chain-finality}).
Then for every PPT adversary
$\mathcal{A}
= \mathcal{A}_{\mathcal{N}}
\cup \mathcal{A}_{\mathcal{C}}
\cup \mathcal{A}_{\mathcal{B}}$,
we have
  $\Pr_{\mathcal{E} \leftarrow \mathcal{A}}[E_{\mathrm{auth}}]
  \;\le\; \epsilon_{\mathrm{chain}}(k)$.
If additionally $\epsilon_{\mathrm{chain}}(k) = e^{-\alpha k}$ for
constant $\alpha > 0$ (exponential finality decay),
then choosing $k = \lceil \alpha^{-1}(\lambda + 1)\ln 2 \rceil$
yields $\Pr[E_{\mathrm{auth}}] \le 2^{-\lambda} = \mathsf{negl}(\lambda)$.
\end{theorem}


\begin{theorem}[Authorization Soundness Violation---Optimistic Execution]
\label{thm:auth-fail}
Consider an x402 system under \emph{optimistic execution},
where $R$ grants access upon $F$'s verification response
before on-chain finality. For any adversary with capability
$\mathbb{A}_{\mathrm{I}} = \{\mathcal{N}, \mathcal{F}, \mathcal{B}\}$,
let $\Delta = T_{\mathrm{verify}} + \delta$ denote the time from
transaction broadcast to service grant, where $\delta \ge 0$ is
adversarial delay on the $F \to R$ path, and let
$T_{\mathrm{inc}}$ be the inclusion time of $tx_{pp}$. If pre-finality
reorganizations occur with probability at least
$p_{\mathrm{reorg}} > 0$ and are independent of HTTP events
(Assumption~\ref{asm:blockchain-indep}), then
  $\mathrm{RGP}_k
  \;\ge\;
  p_{\mathrm{reorg}}
  \cdot
  \Pr\!\bigl[T_{\mathrm{inc}} + k \cdot T_b > \Delta\bigr]$. In particular, when $k = 0$,
$\mathrm{RGP}_0 \ge p_{\mathrm{reorg}} \cdot \Pr[T_{\mathrm{inc}} > \Delta]$.
\end{theorem}


\begin{theorem}[Replay Resistance under Pre-Grant Uniqueness Enforcement]
\label{thm:replay}
Let $\mathcal{S}$ be an x402 system in which $R$ atomically claims each
$(\mathit{pay\_id}, \mathit{resource\_id})$ before granting service,
retains that claim for $\mathrm{TTL}$, and rejects payloads outside a
freshness window $[t-W, t+W]$ with $W < \mathrm{TTL}$. Then for every PPT adversary $\mathcal{A}
= \mathcal{A}_{\mathcal{N}} \cup \mathcal{A}_{\mathcal{C}},$
we have $\forall\, t' \in (t,\, t + \mathrm{TTL}] :
  \quad \Pr[E_{\mathrm{replay}}] = 0$.
Conversely, absent such a pre-grant claim, there exist deployments in
which $\mathcal{A}_{\mathcal{N}}$ replaying $\mathit{PP}$ $n$ times achieves $\mathsf{GrantCount}(\mathit{pay\_id}, \mathcal{E}) = n$ with probability~$1$.

\end{theorem}


\begin{theorem}[Facilitator $k$-Atomicity as Prerequisite]
\label{thm:atomicity}
If $F$ is not $k$-atomic, i.e.,
$\Pr[E_{\mathrm{atom}}] = \mu$ for some non-negligible $\mu$,
then under Assumption~\ref{asm:blockchain-indep}, even conservative execution
cannot achieve authorization soundness: $\Pr[E_{\mathrm{auth}}] \;\ge\; \mu \cdot p_{\mathrm{reorg}}$.
\end{theorem}


\begin{corollary}[Confirmation Depth Recommendation]
\label{cor:depth}
For a target authorization-soundness failure probability
$\epsilon_{\mathrm{target}}$, under conservative execution with
an honest facilitator, the minimum confirmation depth is $k^* = \min\{k : \epsilon_{\mathrm{chain}}(k) \le
  \epsilon_{\mathrm{target}}\}$.
For Base (L2 on Ethereum) with average block time $T_b \approx 2\,\mathrm{s}$ and
$\epsilon_{\mathrm{chain}}(k) \approx e^{-\alpha k}$
for chain-specific $\alpha > 0$, we recommend
$k \ge 3$ for resources valued below \$1 (achieving $\epsilon_{\mathrm{target}} < 10^{-2}$) and
$k \ge 12$ for resources above \$10 (achieving $\epsilon_{\mathrm{target}} < 10^{-4}$).
\end{corollary}


These results clarify how the properties relate. Authorization soundness (Definition~\ref{def:auth-soundness}) is violated by Attack~I when optimistic execution meets reorgs or Byzantine facilitators (Theorem~\ref{thm:auth-fail}). Payment--service correspondence (Definition~\ref{def:no-double}) is violated by Attack~I-B (payment without service) and Attack~II (grants without payment). Replay resistance fails under Attack~II when idempotency is absent (Theorem~\ref{thm:replay}). Facilitator $k$-atomicity is a \emph{prerequisite}: if $F$ reports finality before $k$ confirmations, even conservative servers can be deceived (Theorem~\ref{thm:atomicity}). Since our threat model allows Byzantine facilitators, resource servers that independently verify on-chain finality achieve authorization soundness with respect to $\mathcal{A}_{\mathcal{F}}$.


In the following sections, we instantiate these results through our attacks. We turn the formal failure events into concrete attacks in \S\ref{sec-method}; we then measure the corresponding observables in \S\ref{sec-evalu}; and we examine whether deployed implementations enforce the theorem preconditions in \S\ref{sec:audit}.

\section{Our Attacks}
\label{sec-method}
We now instantiate the security properties from \S\ref{sec:theorems}. Each attack identifies a concrete way to realize one of the formal violation events. Attack~IV is the one exception in scope: it occurs before the payment protocol begins (i.e., the discovery layer) and complements the settlement and replay theorems by showing how an agent can be steered toward a bad x402 endpoint in the first place.


\subsection{Attack I: Settlement-Path Inconsistencies}
\label{sec:attack-i}

x402 decouples \emph{verification} from \emph{settlement}. In the honest path, Client~$C$ requests a paid resource from Resource Server~$R$; $R$ asks Facilitator~$F$ whether the payment looks acceptable; and $F$ later settles against Blockchain~$B$. Attack~I exploits the gap between the HTTP decision at $R$ and the eventual on-chain state at $B$.

Attack~I is the operational form of the authorization and atomicity results. Attack~I-A instantiates Theorem~\ref{thm:auth-fail}: a network- or chain-side adversary widens the gap between ``grant'' and finality, so the server grants before the payment basis is final. Attack~I-B targets a different failure mode: a request-path observer reads the payment header once, races the legitimate facilitator, and consumes a valid authorization without producing the service grant, violating payment-service correspondence.

\subsubsection{\textbf{Attack I-A: Revert-grant under optimistic execution}}\mbox{}\par

\heading{Attack intuition.} The participants are: Client~$C$ (the payer), the resource server $R$
(the party serving resources), Facilitator~$F$ (the verifier and
submitter on the settlement path), and Blockchain~$B$ (the final source
of truth). The honest order is simple: $C$ asks for resources, $R$
checks with $F$, and only then should $R$ release resources.

The vulnerability appears when $R$ optimistically treats a non-final
signal as if it were durable payment. The server may see mempool
presence, shallow inclusion, or a proxy-level ``ok'' from $F$ and
interpret that as finality, even though the payment can still be
replaced away. So the attacker in Attack I-A is not a malformed
client input; it is any network or chain actor that can make the gap
between $R$'s grant and $B$'s finality wide enough to matter.

The attacker may delay or reorder the
network path ($\mathcal{A}_{\mathcal{N}}$), spoof or corrupt the
facilitator ($\mathcal{A}_{\mathcal{F}}$), or influence inclusion,
replacement, and bounded reorganization on the chain
($\mathcal{A}_{\mathcal{B}}$). In short, the attacker needs a way to
make the gap between ``verification looks fine'' and ``settlement is
final'' large enough to exploit.

\smallskip
{\setlength{\fboxsep}{5pt}
\noindent\fcolorbox{black!0}{red!5}{%
\begin{minipage}{0.94\linewidth}
\textcolor{red!65!black}{Failure condition.} A \emph{grant event} $E_{\text{revert}}$ is the moment when
$R$ releases the paid resource to $C$. I-A
succeeds when $R$ grants before the payment reaches the
confirmation threshold, and that same payment later disappears from the
canonical chain.
\end{minipage}}}

\smallskip
Formally, the event is:
\begin{align*}
E_{\text{revert}}=\{g \mid & {\neg\exists\,tx_{pp}}: \mathsf{Final}(tx_{pp},k)\text{ at }t_g \\
&\vee\ tx_{pp}\text{ is later absent from the canonical chain}\}.
\end{align*}

 \begin{figure}[!]
     \centering
    \includegraphics[width=0.9\linewidth]{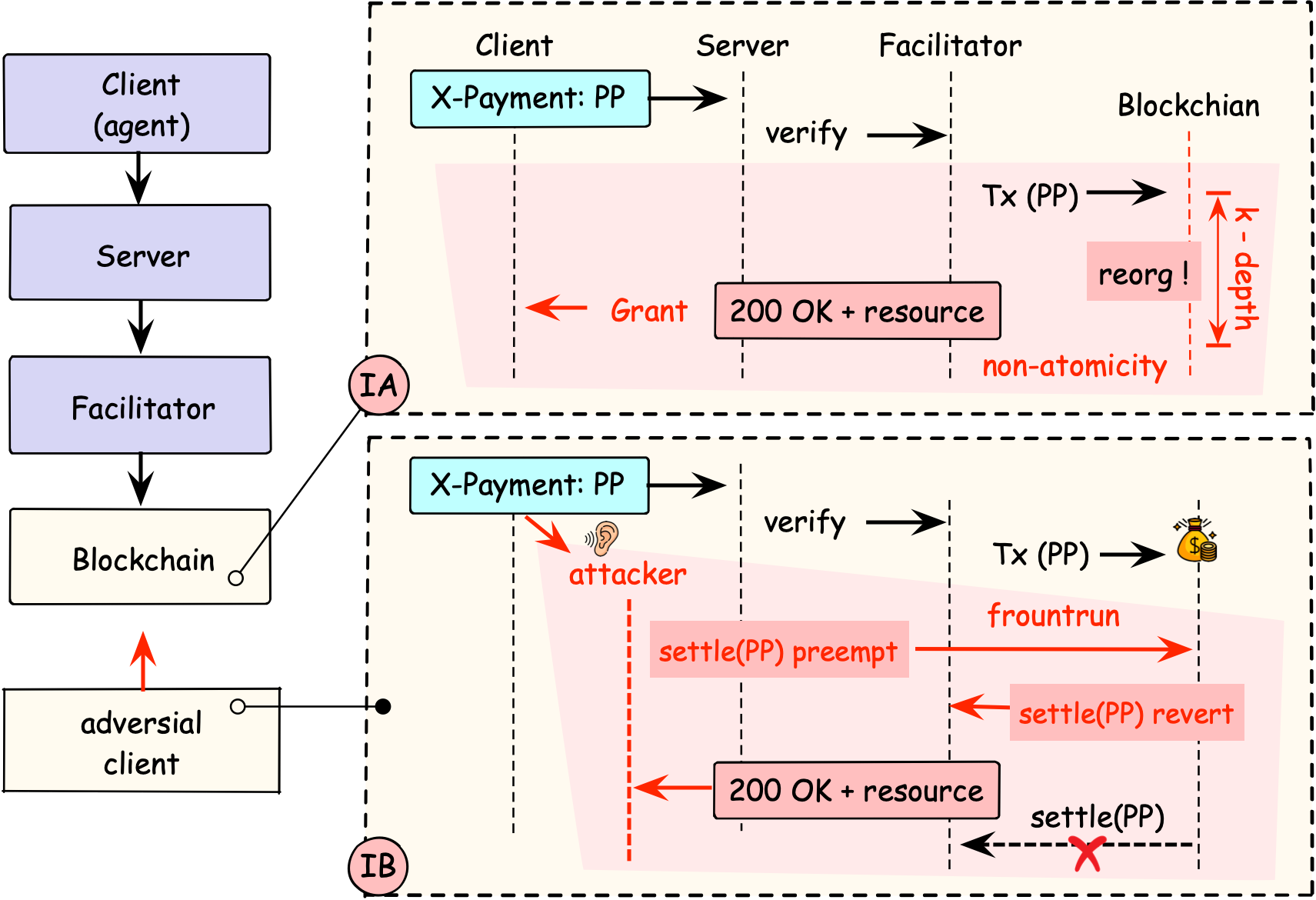}
   \caption{Illustration of Attack~I. In Attack~I-A, the server grants resources before the payment is finalised onchain, so a later reorg may remove the payment. In Attack~I-B, an attacker preempts \texttt{settle(PP)} before the facilitator, so the payment succeeds but the client receives no service.}
\label{fig:attack1}
\vspace{-0.15in}
 \end{figure}

\heading{Attack procedure.} We now walk through the failure in order.

\begin{packeditemize}
\item \textit{Step 1: Client prepares payment.} Client~$C$ requests the
resource from Server~$R$. The server responds with \texttt{402} and a
payment request, after which $C$ sends back the payment header
\texttt{X-PAYMENT}.
\item \textit{Step 2: Server asks for a quick verdict.} Server~$R$
forwards the payment header to Facilitator~$F$ through
\texttt{/verify}. At this point, $F$ may only be checking that the
payment is visible in the mempool or shallowly included.
\item \textit{Step 3: Server grants early.} If $F$ returns
\texttt{ok}, Server~$R$
releases the resource \textcolor{red!70!black}{before}\footnote{We use red to mark attack/failure points throughout this section.} the intended $k$-confirmations. So the key point is: the resource is already gone, but the payment is still reversible.
\item \textit{Step 4: Adversary exploits the gap.} The attacker then
uses the remaining window to prevent durable settlement. Depending on
its position, it can delay the verification path, replace
\textcolor{red!70!black}{$tx_{pp}$ via RBF}, or
trigger a bounded \textcolor{red!70!black}{reorganization} that removes
the transaction from the canonical chain.
\item \textit{Step 5: Server cannot roll the grant back.}
After the reorg or replacement, the payment no longer exists on the canonical chain. But the HTTP response has already been sent, so
$R$ cannot claw the resource back from Client~$C$.
\end{packeditemize}

\heading{Why the attack works.}
Attack~I-A is a classic \emph{time-of-check vs.\ time-of-use} problem: $R$ checks that the payment \emph{looks} acceptable but commits the stronger action of releasing the resource as if settlement were final, and the two moments are not atomic. The pattern mirrors blockchain's ``0-conf acceptance'' risk—a merchant accepts before sufficient confirmations, and a later double-spend or reorg invalidates that acceptance. The key asymmetry is that $\mathcal{B}$ can roll back a pending transfer, but cannot roll back a delivered HTTP response.
The failure is most plausible when an honest facilitator only checks mempool visibility or shallow inclusion while the server treats that reply as final. The facilitator need not be Byzantine; a latency-optimized honest deployment creates the same exposure.





\subsubsection{\textbf{Attack I-B: Unauthorized settlement preemption}}\mbox{}\par

\heading{Attack intuition.} Participants ($C$, $R$, $F$, $\mathcal{B}$) match I-A; what changes is the adversary, a request-path observer or Byzantine server that reads the payment proof before the honest facilitator uses it. The failure is a stolen authorization: the payment is valid but consumed by the wrong caller, so $C$ is charged while $R$ fails to deliver. The attacker can be a network observer ($\mathcal{N}$), e.g., a TLS-terminating proxy or middleware seeing \texttt{X-PAYMENT} in transit, or a Byzantine server ($\mathcal{R}^{*}$) that receives the header and uses it for its own settlement.

This branch uses the same payment object but flips the failure direction: rather than service without durable payment, it produces durable payment without service, the second half of Definition~\ref{def:no-double}.



\vspace{5pt}
{\setlength{\fboxsep}{5pt}
\noindent\fcolorbox{black!0}{red!5}{%
\begin{minipage}{0.94\linewidth}

\textcolor{red!65!black}{Failure condition.} Attack~I-B succeeds when a payment is
consumed on chain, but Server~$R$ never produces the corresponding
service grant. We denote this by $E_{\text{preempt}}$. 
\end{minipage}}}
\vspace{5pt}

The event is the mirror image of replay: replay creates too many grants for one payment, while preemption creates a payment without the expected grant. Formally, the event  $E_{\text{preempt}}$ is:

$E_{\text{preempt}}=\{\exists\,\mathit{PP}: \mathsf{Settlements}.(\mathit{PP})=1 \wedge \textcolor{red!70!black}{\mathsf{GrantCount}}(\mathit{PP})=0\}$

\heading{Attack procedure.} Here the failure is shorter and more familiar than I-A: it is a race for who gets to use the authorization first.

\vspace{3pt}
\begin{packeditemize}
\item \textit{Step 1: Client exposes the payment proof.} Client~$C$
sends \texttt{X-PAYMENT} toward Server~$R$.
\item \textit{Step 2: Attacker copies the authorisation.} A request-path
observer, or a Byzantine server itself, extracts the embedded
authorization from the header; the key is
\textcolor{red!70!black}{copying a usable payment capability}.
\item \textit{Step 3: Attacker races settlement.} Before Facilitator~$F$
submits the honest settlement, the attacker submits the same
authorization from an attacker-controlled account, i.e.,
\textcolor{red!70!black}{settles first}.
\item \textit{Step 4: Honest settlement fails.} Once the nonce is
\textcolor{red!70!black}{consumed}, Facilitator~$F$ no longer has a
usable path to settle the same authorization. Server~$R$ then fails the
request or returns
\texttt{402}, even though Client~$C$ has effectively already paid.
\end{packeditemize}

\heading{Why the attack works.}
The settlement path does not bind the facilitator identity to authorization. In EIP-3009, any observer with the signed authorization can submit the transfer first. In Permit2, a contract without a facilitator field or caller check allows the same preemption. The attacker consumes the nonce, causing the legitimate facilitator's later settlement to fail despite a valid payment intent. Thus, x402 inherits a frontrunning risk whenever payment proofs are exposed and settlement is not bound to the caller.

\textit{Remark.} The risk is especially realistic when payment headers cross logging middleware, TLS terminators, or API gateways, which are components never meant to hold spendable payment material. With caller-unbound settlement, any such observer becomes a potential preemption point.




 \begin{figure}[!]
     \centering
    \includegraphics[width=0.9\linewidth]{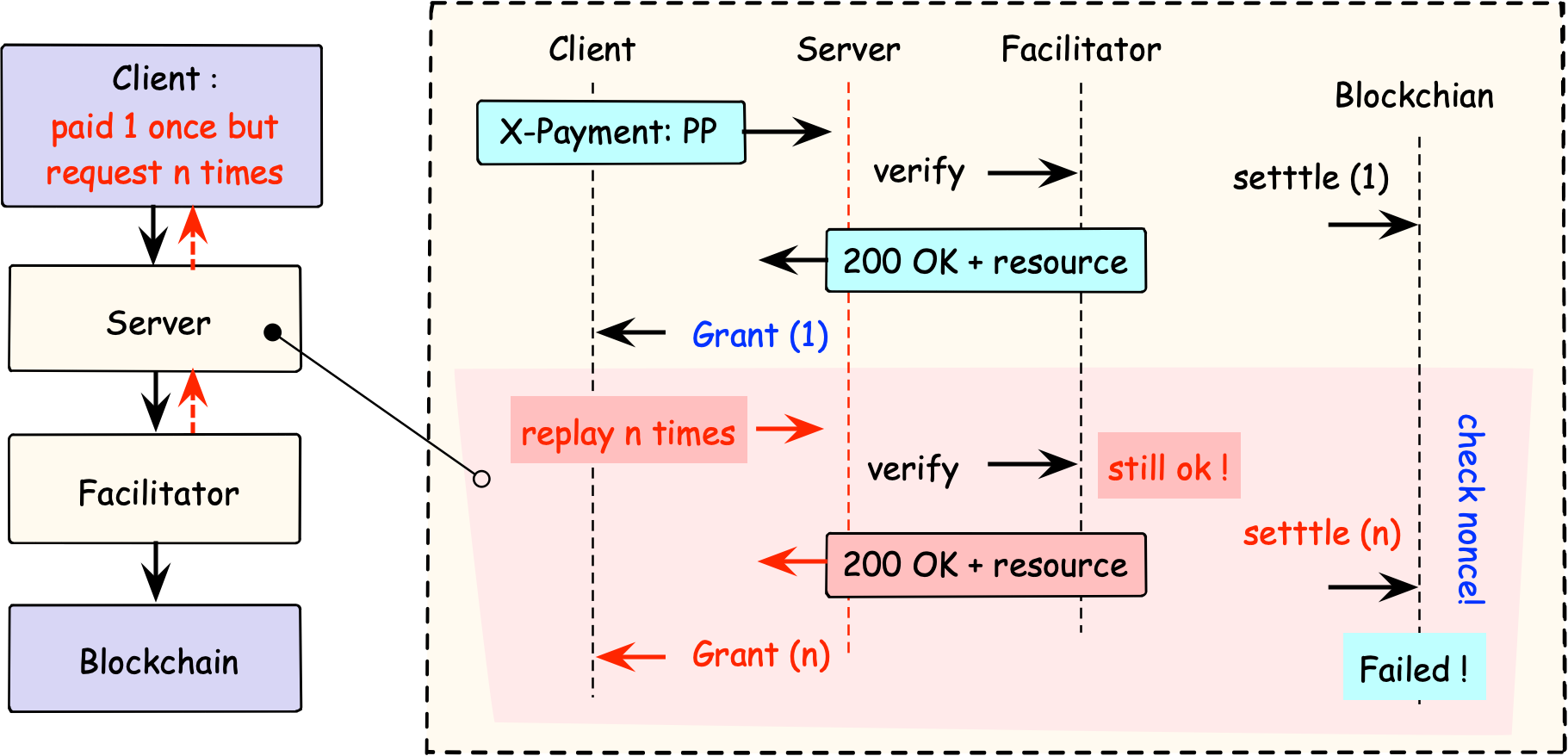}
   \caption{Illustration of Attack~II. A client pays once but replays the same \texttt{X-PAYMENT} payload $n$ times. Without a pre-grant idempotency check, the server may issue $n$ resource grants, while the blockchain accepts only one settlement. }
\label{fig:attack2}
\vspace{-0.1in}
 \end{figure}

\subsection{Attack II: Replay / Idempotency across the HTTP--Chain Boundary}
\label{sec:attack-ii}

Attack~II involves the same four actors as Attack~I: Client~$C$, Server~$R$, Facilitator~$F$, and Blockchain~$B$. Its adversary, however, is usually less dramatic: a replaying client, a retry-happy intermediary, or simply a deployment that fails to record that a payment identity has already been used. 

Attack~II instantiates the converse direction of Theorem~\ref{thm:replay}: unless the server atomically claims $(\mathit{pay\_id},\mathit{resource\_id})$ before granting service, replay becomes a direct payment-service failure, allowing one logical payment to produce multiple HTTP grants.

\heading{Attack intuition.}
In the honest flow, one intended payment buys one protected response; Attack~II breaks that one-to-one mapping. x402 carries a bearer-style payment capability in \texttt{X-PAYMENT}, so if neither $R$ nor $F$ maintains an idempotency record keyed to the logical payment identity, the same capability can be returned through retries, concurrent requests, or semantically equivalent encodings and still appear fresh.

The key asymmetry is that $\mathcal{B}$ may enforce single-use settlement via a nonce or authorization-used bit, while the HTTP side continues to emit grants. The failure is not ``multiple on-chain payments'' but ``multiple off-chain grants for one intended payment'', which is an ordinary web idempotency bug made expensive by attached money. The attacker capability is correspondingly mundane: resend the same payment, or resend it in forms the HTTP layer treats as fresh, until the server grants more than once.




\vspace{3pt}
{\setlength{\fboxsep}{5pt}
\noindent\fcolorbox{black!0}{red!5}{%
\begin{minipage}{0.94\linewidth}
\textcolor{red!65!black}{Failure condition.} A \emph{duplicate-grant
event} $E_{\text{dup}}$ occurs when one logical payment identity leads
to more HTTP-layer grants than the application intended. The payment was
meant to buy one release of resources, but the system releases it
multiple times.
\end{minipage}}}
\vspace{3pt}

In practice, the event is often triggered by ordinary engineering
behavior rather than an obviously malicious client: SDK retries,
load-balancer retries, and concurrent workers can all replay the same
payment if the idempotency store is missing or scoped to the wrong
component. Formally, the event is:
\begin{align*}
E_{\text{dup}}=\{\mathit{pay\_id}\mid {} &
\textcolor{red!70!black}{\mathsf{GrantCount}(\mathit{pay\_id},\mathcal{E})} \\
&\textcolor{red!70!black}{>}
\mathsf{IntendedPayments}(\mathit{pay\_id},\mathcal{E})\}.
\end{align*}

\heading{Attack procedure.}
The failure builds up across layers.
\begin{packeditemize}
\item \textit{Step 1: Client obtains a usable payment header.}
Client~$C$ completes the usual \texttt{402} challenge and prepares a
valid \texttt{X-PAYMENT} value for some payment identity
$\mathit{pay\_id}$.
\item \textit{Step 2: Adversary resends the same payment.} The attacker
\textcolor{red!70!black}{reuses the same \texttt{X-PAYMENT}} across
retries, concurrent requests, or slightly different HTTP encodings that
still refer to the same logical payment.
\item \textit{Step 3: Server checks each replay independently.}
Server~$R$ forwards each replay to Facilitator~$F$ and, without a
deduplication record, treats every arrival as a fresh authorization
attempt.
\item \textit{Step 4: The HTTP layer grants more than once.}
Server~$R$ therefore
releases the resource \textcolor{red!70!black}{again for the same}
logical payment, even though the application intended only one grant.
\item \textit{Step 5: The chain cannot repair the extra grants.}
$B$ may still allow \textcolor{red!70!black}{at most one
successful settlement}, but that invariant arrives too late: the extra
HTTP responses have already been delivered.
\end{packeditemize}

\heading{Why the attack works.}
Attack~II is a classic idempotency failure. Web systems already know
that payment endpoints, callback handlers, and retryable \texttt{POST}
operations need exactly-once protection. x402 inherits the same problem,
except the mistake now sits between a bearer payment header and an
asynchronous on-chain settlement path. The chain's nonce check is not enough, because it protects only contract
state. It does not tell Server~$R$ whether a prior HTTP request with the
same logical payment already caused a resource grant. That is why the
attack survives even when the contract itself behaves correctly: the
contract enforces ``one settlement,'' while the web service still fails
to enforce ``one grant.''

There is also a more subtle version of the same bug: two payloads may be
semantically identical while differing in serialization details such as
JSON ordering, whitespace, or encoding. If a deployment hashes the raw
header bytes instead of a canonical payment identity, replay can hide
inside representation changes rather than byte-for-byte duplicates. So
the attack is not only about obvious duplicates; it is also about the
gap between semantic equality and what the implementation happens to
remember.

\subsection{Attack III: HTTP / Proxy-Level Confusion and Header Manipulation}
\label{sec:attack-III}

Attack~III moves one layer earlier, to the HTTP path itself. The key participants are Client~$C$, an intermediary~$P$ such as a reverse proxy or CDN edge, and Server~$R$. The adversary is not a chain manipulator, but an intermediary or middleware component that can normalize, rewrite, merge, cache, or log traffic before the origin server observes it.  Attack~III carries the payment--service property into the surrounding HTTP infrastructure: even when on-chain settlement is correct, proxy rewriting or shared caching can make the service boundary diverge from the payment boundary.

Unlike Attack~I and Attack~II, the core problem here is not whether the
payment is valid on chain. The problem is that x402 places spendable
payment material inside ordinary HTTP machinery. Once that happens, the
payment header and the paid response inherit the same parsing
ambiguities and caching hazards that already plague authenticated web
traffic. The two most important failures are different enough that it is
better to separate them: (\textit{D1}) one changes the header in flight, and (\textit{D2}) the
other leaves the payment alone but leaks the paid response later.

 \begin{figure}[!]
     \centering
    \includegraphics[width=0.99\linewidth]{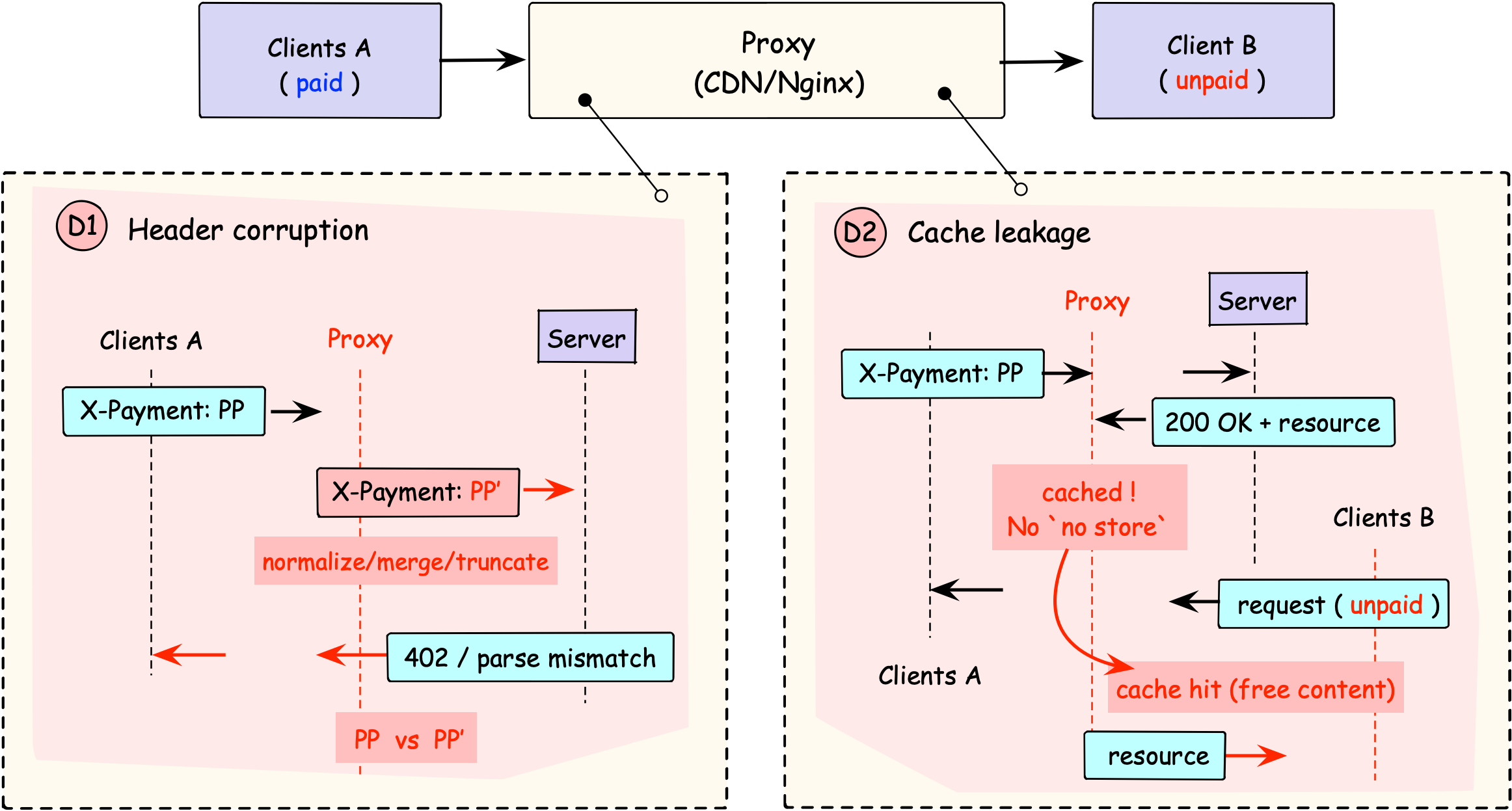}
   \caption{Illustration of Attack~III. A client request passes through an HTTP proxy before reaching the resource server. In D1, the proxy rewrites \texttt{X-PAYMENT}, causing a payload mismatch. In D2, the proxy caches a paid response without \texttt{no-store}, allowing a later unpaid client to receive it for free.}
\label{fig:attack3}
\vspace{-0.1in}
 \end{figure}

\heading{Attack intuition} (Fig.\ref{fig:attack3}).
The participants are paying Client~$C$, intermediary~$P$, Server~$R$,
and sometimes a later unpaid client $C'$. Attack~III is still one
attack class: an HTTP intermediary distorts the payment boundary before
the origin has a clean chance to enforce it. In one manifestation,
Proxy~$P$ changes the \texttt{X-PAYMENT} header in flight, so
Server~$R$ reasons about different payment semantics than Client~$C$
intended. In the other, the request is handled correctly, but the paid
response is later replayed because an intermediary caches it as if it
were ordinary public content.

These two outcomes have a different flavor. Header corruption (\textit{D1}) is a
parser-risk that depends on how multiple components normalize and
interpret the same request. Cache leakage (\textit{D2}) is cleaner and more concrete: the payment may be perfectly valid, yet the paid artifact later reaches
someone who never paid. They belong together because both arise from the
same design move, namely carrying spendable payment material and paid
artifacts through ordinary HTTP infrastructure.

In our evaluation, cache leakage is the primary validated outcome, while
header ambiguity is better understood as a realistic parser-risk that
appears once real proxies, CDNs, API gateways, or observability layers
are inserted in front of the origin.

\vspace{5pt}
{\setlength{\fboxsep}{5pt}
\noindent\fcolorbox{black!0}{red!5}{%
\begin{minipage}{0.94\linewidth}
\textcolor{red!65!black}{Failure condition.} Attack~III succeeds when an
intermediary either causes Server~$R$ to see payment semantics different
from what $C$ sent, or replays a paid response outside the
original paid context so that an unpaid client receives it.
\end{minipage}}}
\vspace{5pt}

Formally, we capture the attack as the union of two 
events:
\begin{align*}
E_{\text{III}} = {} &
\left\{req \mid
\textcolor{red!70!black}{\mathsf{Sem}_{R}(\texttt{X-PAYMENT})}
\textcolor{red!70!black}{\neq}
\mathsf{Sem}_{C}(\texttt{X-PAYMENT})\right\} \\
& \cup
\left\{req' \mid
\mathsf{ServedFromCache}(req')
\wedge
\textcolor{red!70!black}{\neg \mathsf{Paid}(req')}\right\}.
\end{align*}

\heading{Attack procedure.}
The exact path depends on the intermediary, but the overall failure is
simple.
\begin{packeditemize}
\item \textit{Step 1: Client sends a paid request.} Client~$C$ emits a
request carrying \texttt{X-PAYMENT} toward Server~$R$ through
intermediary~$P$.
\item \textit{Step 2: The intermediary handles payment-sensitive HTTP
state as ordinary traffic.} Depending on the deployment, Proxy~$P$ may
\textcolor{red!70!black}{normalize}, \textcolor{red!70!black}{merge}, or \textcolor{red!70!black}{truncate} the payment
header, or it may later \textcolor{red!70!black}{cache} the paid
response because the application omitted strong directives such as
\texttt{no-store}.
\item \textit{Step 3: The wrong artifact becomes authoritative.}
Server~$R$ may verify a mutated header \textit{PP}$'$ rather than the
client's original \textit{PP}, or a later unpaid request may be
answered from a stored copy of a previously paid response.
\item \textit{Step 4: The payment boundary collapses.} The system either
acts on a parser view that \textcolor{red!70!black}{changed in transit}
or returns paid content to an \textcolor{red!70!black}{unpaid client}.
\end{packeditemize}

\heading{Why the attack works.}
Attack~III is closest in spirit to classic request-confusion and authenticated-response cache bugs: HTTP components share no canonical interpretation of transformed headers, and shared caches do not natively understand that a response was produced only after a payment gate—they simply follow local semantics. The x402-specific twist is that the ambiguous object is not just an authentication token but spendable payment material or a paid artifact, which raises the stakes of familiar web behaviors such as duplicate-header merging, canonicalization, truncation, logging, and edge caching, all of which become payment-relevant once they sit on the path.

The two manifestations differ: header corruption is a deployment-dependent parser risk where components disagree on the verified payment object, while cache leakage is a cleaner end-to-end exploit—payment settles correctly, yet the resource is later replayed for free because the cache boundary ignored the payment boundary.




\subsection{Attack IV: Server-Selection Attacks}
\label{sec:attack-IV}

Attack~IV extends the model upstream of $\Pi_{x402}$, from settlement to discovery. While the formal theorems assume a chosen server and analyze whether its payment flow is safe, Attack~IV shows that agent-facing discovery can bias which server receives the payment opportunity in the first place. The relevant participants are a discovery layer such as Bazaar, an LLM agent that selects among candidate servers, honest servers $R_h$, and attacker-controlled servers $R^\ast$. The adversary is a malicious seller that need not break cryptography or settlement logic; it only needs to manipulate how the agent perceives, ranks, and selects available servers.

 \begin{figure}[H]
     \centering
    \includegraphics[width=0.99\linewidth]{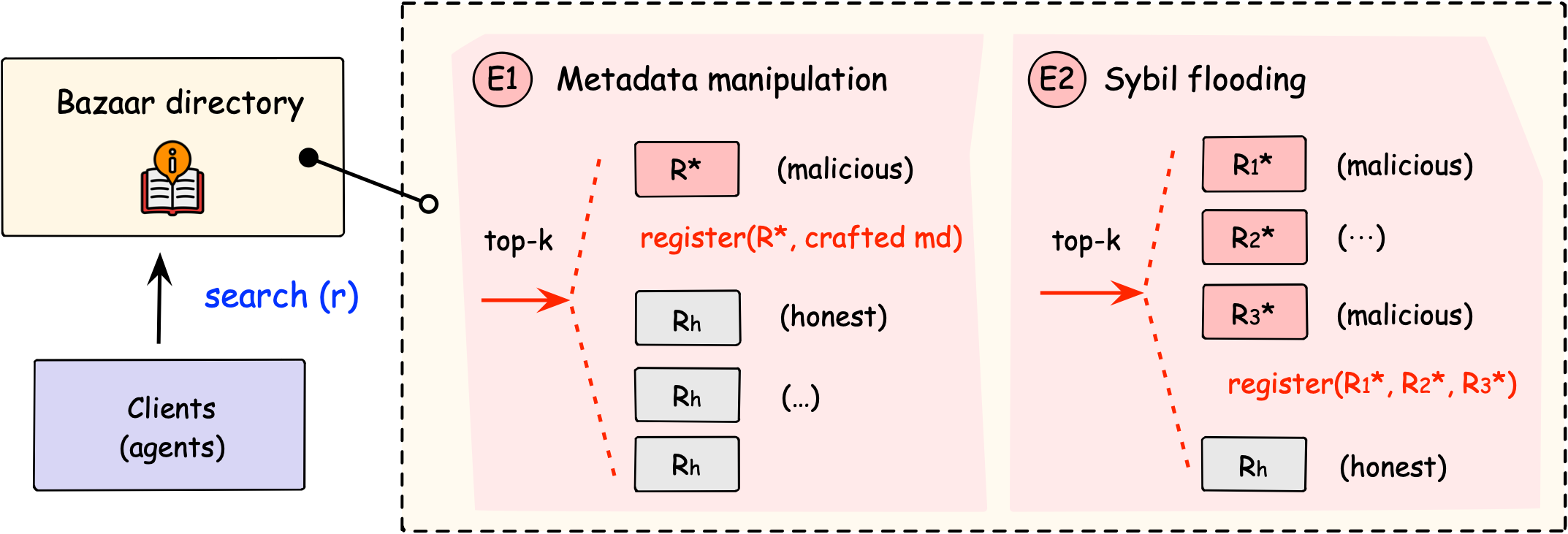}
   \caption{Illustration of Attack~IV.  The attacker, in E1, crafts metadata to place a malicious server in shortlist; in E2, registers multiple Sybil servers to crowd out honest providers.}
\label{fig:attack4}
\vspace{-0.1in}
 \end{figure}

\heading{Attack intuition} (Fig.\ref{fig:attack4}).
In Bazaar-style discovery, the agent does not manually inspect the whole
market. It queries a catalog, receives a shortlist, and then chooses one
server based on names, descriptions, prompts, prices, trust cues, or
latency cues. Attack~IV exploits that decision pipeline. If an attacker
can shape the metadata or multiply its presence through Sybil listings,
the agent can be nudged toward an attacker-controlled paid endpoint
before it ever reasons carefully about the market as a whole.

This matters more in x402 than in an ordinary tool-selection mistake because discovery is directly connected to payment. Once the agent selects a server, that server can immediately issue payment requests and influence the agent's next action. The attacker does not need to compromise the wallet or the settlement contract; it only needs to win the selection step often enough.
The harm is far beyond poor service choice but on-chain economic loss.

\vspace{5pt}
{\setlength{\fboxsep}{5pt}
\noindent\fcolorbox{black!0}{red!5}{%
\begin{minipage}{0.94\linewidth}
\textcolor{red!65!black}{Failure condition.} A
\emph{selection-bias event} occurs when an attacker-controlled server is
chosen more often than the baseline market would justify after discovery
has been manipulated.
\end{minipage}}}
\vspace{5pt}

Let $\mathcal{R}_k$ denote the shortlist returned to the agent for a
query~$q$, and let $\pi_\Delta(R^\ast \mid q,\mathcal{R}_k)$ be the
probability that the agent selects the adversarial server under
manipulation~$\Delta$. We capture the success event as
$E_{\text{sel}}=\{q \mid
 \textcolor{red!70!black}{\pi_\Delta(R^\ast \mid q,\mathcal{R}_k)}
 \textcolor{red!70!black}{>}
 \pi_0(R^\ast \mid q,\mathcal{R}_k)\}$.

\heading{Core manipulations.}
The cleanest forms of Attack~IV are the first two. In \textit{E1:
metadata manipulation}, the attacker rewrites names, descriptions, and
LLM-facing prompts so that $R^\ast$ appears more relevant than honest
competitors. In \textit{E2: Sybil flooding}, the attacker registers
multiple aliases $R^\ast_1,\dots,R^\ast_r$ so that adversarial entries
occupy more of the shortlist returned by discovery.

The remaining variants are better read as extensions than as the core
story: \textit{E3} manipulates price and payment signals, \textit{E4}
leans on latency and availability cues, and \textit{E5} exploits how
facilitator affiliation or payment-path metadata is presented to the
agent. They matter, but they all build on the same underlying move:
shape the shortlist first, then let the agent do the rest.

Among the five variants, E1 and E2 are the cleanest core attacks as
they require only ordinary seller-side registration powers. E3-E5 are
best read as extensions that manipulate richer marketplace signals after
the attacker has already gained a place in the shortlist.

\heading{Attack procedure.}
The manipulation unfolds before the victim ever sends a payment.
\begin{packeditemize}
\item \textit{Step 1: Adversary prepares the listing.} The attacker
\textcolor{red!70!black}{registers crafted metadata} for one malicious
server or several Sybil aliases.
\item \textit{Step 2: Discovery returns a shortlist.} For a target
query~$q$, the discovery service returns top-$k$ candidate
servers to the agent.
\item \textit{Step 3: The shortlist is biased.} Since the adversary
has tuned relevance, quantity, or trust cues, the malicious entries occupy \textcolor{red!70!black}{more of the visible search surface} than under a clean baseline.
\item \textit{Step 4: The agent picks the wrong server.} The agent then selects \textcolor{red!70!black}{the attacker-controlled server}
$R^\ast$ and routes traffic for payments.
\item \textit{Step 5: Economic harm follows selection.} Once discovery
has already steered the workflow toward $R^\ast$, the payment path and
service interaction proceed on top of a manipulated choice.
\end{packeditemize}

\heading{Why the attack works.}
Attack~IV resembles tool-selection and ranking attacks on LLM agents: discovery compresses the catalog into a shortlist, then the model decides from that filtered view. Once the attacker enters the shortlist with persuasive metadata, the decision is often lost before deeper reasoning begins. 

The x402 setting adds two amplifiers. First, Bazaar-like ecosystems often have weak registration controls, so adversaries can cheaply create or rename listings until one lands in the right lexical neighborhood. Second, the chosen server is monetized, so a bad selection translates directly into payment requests, poor service, or financial loss. Sybils compound this by raising the chance that at least one adversarial listing wins. Attack~IV is thus a market-surface attack: the adversary need not outperform every honest server, only be more visible or persuasive when the shortlist is formed.




\section{Evaluation}
\label{sec-evalu}

This section shows our experimental design and results. We first summarize the shared configuration settings that apply across the evaluation. We then evaluate each attack class in turn.

The evaluation follows the violation events defined in \S\ref{sec:security-props}. RGP$_k$ measures authorization-soundness failures (Definition~\ref{def:auth-soundness}) under optimistic or non-atomic settlement; DGR measures replay payment-service failures (Definition~\ref{def:no-double} and Theorem~\ref{thm:replay}); cache leakage captures violations at the HTTP boundary (Definition~\ref{def:no-double}); and selection rate captures the discovery-layer precursor that determines which x402 server receives the payment flow (\S\ref{sec:attack-IV}).

\subsection{Shared Configuration Settings}

\heading{System components.} Our testbed consists of four components:

\begin{packeditemize}
\item \textit{Resource Server $R$} (Node.js~20 / Express): implements the full x402 payment flow, including the initial \texttt{402} challenge, \texttt{X-PAYMENT} receipt,  
\texttt{/verify} call, and access decision. We configure $R$ in an optimistic mode (grant soon after settlement submission) or a conservative mode (wait for explicit confirmation depth $k$).

\item \textit{Facilitator $F$} (Node.js / Express / ethers.js~v6): verifies signatures and settles, supporting \emph{honest} (wait for receipt), \emph{optimistic-bug} (reply ``ok'' on mempool or shallow inclusion), and \emph{byzantine} (return ``ok'' without on-chain enforcement) modes.


\item \textit{Blockchain $\mathcal{B}$}: a local Hardhat chain for controlled reorg simulation and Base Sepolia for realistic latency and live endpoint validation for attack surfaces.


\item \textit{Client harness $\mathcal{C}$}: issues concurrent x402 flows, replays headers, and drives proxy and discovery experiments.
\end{packeditemize}

\heading{Hardware configurations.}
All local experiments run on a single Apple M3 Max workstation with a 16-core CPU under macOS.

\heading{Instrumentation.} 
$R$ records per-grant metadata (\textit{pay\_id}, tx hash, grant time, policy, block number); $F$ logs verify calls, settlement attempts, duplicates, and reverts. Outputs are serialized to timestamped JSON traces for reproducibility.


\heading{Reorg and delay injection.}
On the local chain, a reorg injector snapshots state, mines a branch containing $tx_{pp}$, and reverts with probability $p_{\mathrm{reorg}}$. For larger sweeps, we use an analytic model with a Bernoulli revert trial with probability $p_{\mathrm{reorg}}$ at $k=0$ and $p_{\mathrm{reorg}}^{k}$ for $k>0$. Application-level delay $\delta \in \{0,100,200,400\}\,\mathrm{ms}$ is injected on the $R \rightarrow F$ path. Attack~I-B is evaluated separately via Base Sepolia preemption traces. For Attacks~II and~III, the harness replays headers under controlled concurrency through intermediaries that may mutate or cache content.



\heading{Parameter space.} Attacks~I-A/B vary $k \in \{0,3,6,12\}$, $p_{\mathrm{reorg}} \in \{0,0.01,0.05\}$, $\delta \in \{0,100,200,400\}\,\mathrm{ms}$, facilitator honesty, and execution policy. Attack~II varies replay count $n \in \{1,5,10,50\}$ and idempotency keying. Attack~III runs ${\sim}1{,}000$ requests per condition through nginx, Caddy, and a custom MitM proxy. Attack~IV evaluates 12 categories $\times$ 15 queries across three LLMs with Sybil count $r \in \{1,3,5\}$. Local state is reset between configurations.


\heading{External systems and artifacts.}
To avoid overfitting to a purely local stack, we also evaluate four public x402-enabled testnet endpoints across different application categories: blockchain data, general infrastructure, search/analytics, and domain-specific APIs. They are anonymized in the submission version for coordinated-disclosure reasons.  We also audit additional x402 SDKs in TypeScript, Python, and Rust to compare replay handling, settlement binding, and header processing across codebases.

\heading{Smart contract and statistics.}
Our controlled on-chain experiments use a MockUSDC contract that faithfully implements EIP-3009
\texttt{transferWithAuthorization}, including per-authorizer nonce tracking, EIP-712 signing, and authorization-used events. For rate metrics such as RGP$_k$, DGR, cache leakage incidence, and selection rate, we report point estimates with 95\% Wilson score confidence intervals using $z=1.96$. Latency metrics such as $T_{\mathrm{gf}}$ are summarized using the median and interquartile range (IQR).

\subsection{Evaluating Attacks I-A and I-B}
\label{sec:eval-attack1}




\heading{Experimental setup.} 
For Attack~I-A, we sweep confirmation depth $k \in \{0,3,6,12\}$, reorg probability $p_{\mathrm{reorg}} \in \{0,0.01,0.05\}$, and delay $\delta \in \{0,100,200,400\}\,\mathrm{ms}$ on the local chain, comparing optimistic vs.\ conservative execution and honest vs.\ Byzantine facilitators. Each run uses a fresh EIP-3009 authorization with the chain reset between configurations, and we retain one end-to-end trace illustrating the revert-grant path. Base Sepolia serves as a timing sanity check for the grant-to-finality gap but is not used to estimate revert probability, since public-testnet reorgs cannot be injected.

For Attack~I-B, we validate two settlement paths on Base Sepolia: EIP-3009 preemption against a live endpoint settling via \texttt{transferWithAuthorization}, and a Permit2 PoC against the \texttt{x402ExactPermit2Proxy} contract. In both, the adversary observes \texttt{X-PAYMENT} before the honest facilitator settles.

\heading{Metrics.} For Attack~I-A we measure RGP$_k$, the probability of a revert-grant
event at depth~$k$, and $T_{\mathrm{gf}}$, the time gap between the HTTP grant and durable chain finality. For Attack~I-B we record whether an attacker can consume the authorization before the expected facilitator path settles, yielding payment without service. This is intentionally a binary validation metric rather than a broad rate estimate: the point of
Attack~I-B is to show that the settlement path is caller-unbound and therefore preemptable at all.

\heading{Results on Attack~I-A} (Fig.\ref{fig:attack1-results}).
Attack~I-A shows a security-latency tradeoff. In the local sweep, revert-grant events occur when the server grants on a non-final signal and reorg probability is non-negligible. At $k=0$ in the fixed $T_b=2$\,s main sweep, RGP$_0$ reaches 4.70\% under $p_{\mathrm{reorg}}=0.05$ and $\delta=200$\,ms, and rises to 5.18\% when the injected delay increases to 400\,ms. Increasing $k$ reduces observed risk to the CI floor, but raises $T_{\mathrm{gf}}$ to tens of seconds. Stronger confirmation mitigates Attack~I-A at the cost of user-visible latency.

\begin{figure}[H]
  \centering
  \includegraphics[width=\linewidth]{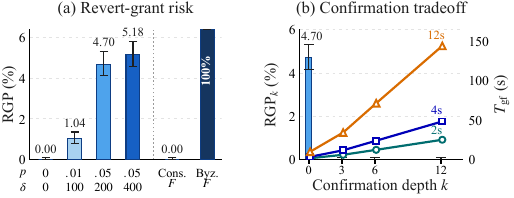}
  \caption{Attack~I-A results with $T_b{=}2$\,s and 5{,}000 requests per condition. 
(a) RGP$_0$ grows with reorg probability and delay. 
(b) Increasing $k$ lowers RGP$_k$ to the CI floor but raises $T_{\mathrm{gf}}$ from 1.6\,s to 25.1\,s.
}
  \label{fig:attack1-results}
  \vspace{-0.1in}
\end{figure}

Control cases bound the behavior: conservative execution with an honest facilitator yields no revert-grants, while a Byzantine facilitator reaches the upper bound. The local trace reproduces the failure path: $C$ pays, $F$ approves, $R$ grants before finality, and an analytic reorg removes $tx_{pp}$. Base Sepolia confirms that the grant-to-finality gap exists outside the lab, but cannot estimate revert probabilities because public-testnet reorgs cannot be injected.

\heading{Results on Attack~I-B.}
Attack~I-B is a feasibility validation rather than a rate estimate. The live EIP-3009 trace exercises the full HTTP path on an anonymized Base Sepolia endpoint. After the payer's \texttt{X-PAYMENT} header is exposed, the same authorization is submitted directly to testnet USDC through \texttt{transferWithAuthorization}; the on-chain transfer succeeds in \href{https://sepolia.basescan.org/tx/0x32fb998ba4647c5c891852a21232bb4a9842b51b1f400f91faf32317e06e556d}{tx~\texttt{0x32fb\ldots{}e556d}}, consumes the nonce, and charges the payer 0.0001~testnet USDC. The endpoint's later settlement attempt then fails on the consumed nonce, and the resource request returns HTTP~402.

We observe the same caller-unbound pattern in a controlled Permit2 proof of concept against \href{https://sepolia.basescan.org/address/0x402085c248EeA27D92E8b30b2C58ed07f9E20001}{\texttt{x402ExactPermit2Proxy}} on Base Sepolia. An unrelated attacker EOA calls \texttt{settle()} with the observed Permit2 signature, and the settlement succeeds in \href{https://sepolia.basescan.org/tx/0x3f3046cb96a738c71a90120445cba03685a6a22fdc3ece6e93883d8bfc5793f1}{tx~\texttt{0x3f30\ldots{}93f1}}, consuming the nonce. A later legitimate replay reverts because the nonce has already been used. Together, the EIP-3009 and Permit2 traces isolate the same preemption condition: an attacker consumes the payment capability before the expected settlement path can use it. The EIP-3009 trace shows the paid-but-denied HTTP outcome, while the Permit2 trace confirms the same nonce-consumption failure at the settlement contract.

\vspace{3pt}
{\setlength{\fboxsep}{5pt}
\noindent\fcolorbox{black!0}{orange!10}{%
\begin{minipage}{0.94\linewidth}
\textcolor{red!65!black}{Takeaways.} 
x402 safety depends on deployment choices: it requires either extra latency or stricter settlement binding. 
\end{minipage}}}

\subsection{Evaluating Attack II}
\label{sec:eval-attack2}

\heading{Experimental setup.}
We replay the same payment payload concurrently with replay count $n \in \{1,5,10,50\}$. Each request carries the same logical payment identity: identical $\mathit{pay\_id}$, nonce, signature, and base64-encoded \texttt{X-PAYMENT} header. The local trials use real settlement on MockUSDC through \texttt{transferWithAuthorization}, so each authorization can settle at most once. We run the same matrix without idempotency and with idempotency keyed by $\mathit{pay\_id}$.

We then test a live Base Sepolia endpoint, Endpoint-2, using a barrier-released batch of 1{,}000 concurrent replays of one payment header. This deployment uses \texttt{PAYMENT-SIGNATURE}, an \texttt{X-PAYMENT}-equivalent payment header. We inspect Python and Rust SDKs to identify implementation choices that may widen the replay window.

\heading{Metrics.}
The primary metric is DGR, the number of HTTP grants per intended payment identity. For live testnet replay, we record the on-chain settlement count to separate duplicate grants from duplicate payments. Grants are counted as HTTP~200 responses. Settlement count is estimated from USDC balance delta divided by the per-request price (\$0.001). Response hashes and balance deltas are logged to distinguish HTTP grants from on-chain settlements.



\heading{Results on Attack~II} (Table~\ref{tab:attack2-results}).
The local replay matrix shows the core mismatch: the chain enforces one settlement, but the HTTP layer may grant multiple times. Without $\mathit{pay\_id}$ deduplication, every replay produces another grant, giving DGR$\,{=}\,n$. With idempotency at the HTTP boundary, duplicate grants disappear. Thus, on-chain nonce discipline does not by itself provide exactly-once service semantics off-chain.

\begin{table}[t]
\centering
\caption{Attack II: HTTP grants when replaying the same \texttt{X-PAYMENT} $n$ times. Each cell shows the number of HTTP-layer grants (= DGR) per single intended payment.}
\label{tab:attack2-results}
\small
\begin{tabular}{lcccc}
\toprule
Implementation & $n{=}5$ & $n{=}10$ & $n{=}50$ & Idempotency \\
\midrule
Testbed (no check)
  & 5 & 10 & 50
  & \xmark \\
Testbed ($\mathit{pay\_id}$)
  & 0 & 0 & 0
  & \cmark \\
SDK-B (Python)
  & 5 & 10 & 50
  & \xmark \\
SDK-C optimistic
  & 5 & 10 & 50
  & \xmark \\
SDK-C pessimistic
  & 1 & 1 & 1
  & via settle \\
\bottomrule
\end{tabular}
\end{table}

The SDK evidence shows that implementation choices can widen this gap. SDK-B has grant-before-settle behavior: its streaming path flushes the paid response before \texttt{settle\_payment} runs, so the client can receive the service even when settlement later fails. SDK-C depends on mode. In optimistic mode, it dispatches settlement via \texttt{tokio::spawn} and returns immediately, leaving a replay window in which repeated requests can pass verification and reach DGR$\,{=}\,n$. In pessimistic mode, settlement occurs synchronously before the response, so duplicate nonces revert and DGR drops to~1.

The live testnet replay confirms the duplicate-grant effect under a deployed endpoint. On Endpoint-2, we replay one payment header across 1{,}000 concurrent requests and measure both grants and USDC settlements. In the strongest positive round, we observe 248 HTTP-layer grants and 1 on-chain settlement. This establishes duplicate granting at the HTTP layer while the on-chain payment is consumed only once.

\vspace{3pt}
{\setlength{\fboxsep}{5pt}
\noindent\fcolorbox{black!0}{orange!10}{%
\begin{minipage}{0.94\linewidth}
\textcolor{red!65!black}{Takeaways.} 
Attack~II does not break settlement: the on-chain authorization is consumed only once. The failure is at the HTTP layer: if the server does not remember used payments (i.e., idempotency), the same payment can trigger multiple grants.
\end{minipage}}}

\subsection{Evaluating Attack III}
\label{sec:eval-attack3}

\heading{Experimental setup.}
We insert intermediary layers between the client and $R$ to test header mutation and cache leakage separately. The local harness uses an Express origin that logs headers and mimics paid 200 responses with \texttt{PAYMENT-RESPONSE} but no default \texttt{Cache-Control}. We test nginx with \texttt{proxy\_cache}, default Caddy, and a custom MitM proxy that injects duplicate payment headers. Header tests use unique URLs to force origin hits. Cache tests reuse one URL and restart proxies between conditions for a cold cache.


We also audit public x402 testnet endpoints settling via Base Sepolia USDC: we record the \texttt{Cache-Control} on a paid response, replay the payment signature five times, and issue an unpaid request for the same resource.

\heading{Metrics.}
We measure whether the payment header changes in transit (Header Mutation Rate) and whether a paid response is later served to an unpaid client (Cache Leakage Incidence). For duplicate-header injection, we also record how the downstream stack resolves multiple payment-header values, since this captures parser ambiguity more directly than a standalone count.

\begin{table}[H]
\centering
\caption{Attack III: proxy and caching behavior across middleware stacks. 1{,}000~requests per proxy per test.}
\label{tab:attack3-results}
\small
\resizebox{\linewidth}{!}{
\begin{tabular}{lccc}
\toprule
\multicolumn{1}{c}{\textbf{Proxy}} & \textbf{Header Mutation}
  & \makecell{\textbf{C. Leak} (no CC)} & \makecell{\textbf{C. Leak} (with CC)} \\
\midrule
nginx (Docker)
  & 0.0\% & \textbf{100.0}\% & 0.0\% \\
Caddy (Docker)
  & 0.0\% & 0.0\% & 0.0\% \\
Custom MitM
  & \textbf{100.0}\% (multi-header) & --- & --- \\
\bottomrule
\end{tabular}
}
\end{table}

\begin{figure*}[t]
  \centering
  \includegraphics[width=0.95\textwidth]{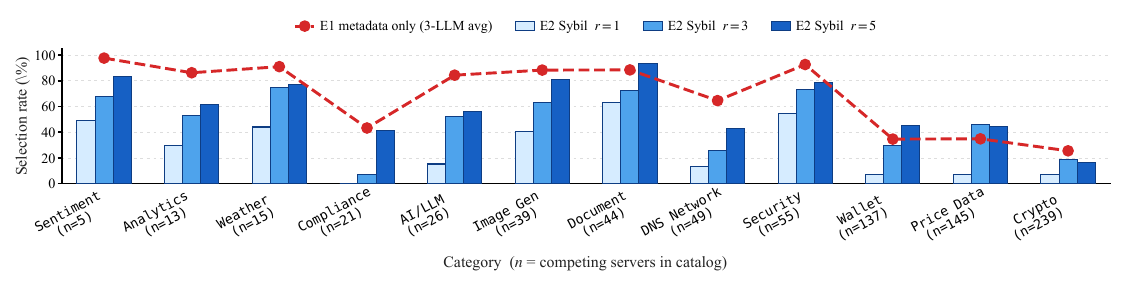}
  \caption{Attack~IV selection rates across 12 categories ordered
    by competition density. Bars: E2 Sybil flooding at
    $r{\in}\{1,3,5\}$ (3-LLM avg). Red dashed line: E1
    metadata-only baseline. $n$ = competing servers per category.}
  \label{fig:attack4-e1e2}
\end{figure*}

\heading{Results on Attack~III} (Table~\ref{tab:attack3-results}).
The local proxy tests show that the two forms of Attack~III differ in stability. Header corruption is possible but configuration-dependent. 
Neither nginx nor Caddy rewrites \texttt{X-PAYMENT} in our setup (0\% mutation across 2{,}000 requests), while a custom MitM proxy injects a duplicate \texttt{X-PAYMENT} in 100\% of 1{,}000 requests, with the Node/Express stack consistently exposing the last value.
This establishes parser ambiguity but does not itself constitute an end-to-end bypass.

Cache leakage is direct. With nginx \texttt{proxy\_cache} enabled and no \texttt{Cache-Control} protection, all 1{,}000 unpaid requests receive the cached paid content. When the upstream emits \texttt{Cache-Control: no-store, private} from a cold cache, 0/1{,}000 unpaid requests are served from cache. Caddy does not cache in either condition.

The live audit shows the same risk. Endpoint-1 returned publicly cacheable paid responses, allowing unpaid clients to retrieve cached content during the CDN cache window. More broadly, audited server-side SDKs do not automatically add \texttt{Cache-Control: no-store} or \texttt{private} to payment-gated responses, leaving cache safety to application or CDN configuration.

\vspace{5pt}
{\setlength{\fboxsep}{3pt}
\noindent\fcolorbox{black!0}{orange!10}{%
\begin{minipage}{0.94\linewidth}
\textcolor{red!65!black}{Takeaways.} 
The operational risk is not header mutation but cache state escaping the payment boundary. Once paid content is publicly cacheable, later clients can retrieve it without paying.
\end{minipage}}}

\subsection{Evaluating Attack IV}
\label{sec:eval-attack4}

\heading{Experimental setup.}
We emulate \texttt{x402\_discover} using a public x402 endpoint\footnote{\url{https://x402scout.com/}} catalog. The retriever matches queries against endpoint names, descriptions, and tags, uses trust score as a tiebreaker, and returns a top-10 shortlist. This approximates Bazaar-style discovery while letting adversarial entries surface by normal ranking.

We test 12 categories with 15 queries each across three LLMs: MiniMax-M2.7, GPT-5.3, and Sonnet~4.5, all at temperature~0.1. In E1, the attacker adds one crafted server per category. In E2, the attacker adds $r \in \{1,3,5\}$ Sybil servers per category, each with a distinct domain and realistic attributes. Overall, we collect 540 LLM decisions for E1 and 1{,}620 for E2, totaling 2{,}160 discovery decisions.

\heading{Metrics.}
The primary metric is selection rate, defined as the fraction of valid model decisions for which an attacker-controlled server is chosen. We also track applicability, since selection requires the malicious listing to appear in the returned shortlist. For E2, we compare aggregate selection rate as the Sybil count increases from $r{=}1$ to $r{=}5$.

\heading{Results on Attack~IV} (Fig.\ref{fig:attack4-e1e2}).
Discovery can be captured before payment execution begins. The attacker does not need to dominate the full catalog; it only needs to enter the shortlist used by the agent. In E1, a single attacker server with crafted metadata achieves 71.8\% selection on MiniMax-M2.7, 69.4\% on GPT-5.3, and 68.8\% on Sonnet~4.5 across all 12 categories. The similar pattern across models suggests that the effect mainly comes from the discovery setup rather than a single model. Table~\ref{tab:attack4-e2-full} breaks the E2 Sybil results down by model.

\begin{table}[H]
\centering
\caption{Attack IV E2 per-model breakdown: Sybil selection
rate~(\%) by model and Sybil count~$r$.}
\label{tab:attack4-e2-full}
\footnotesize
\setlength{\tabcolsep}{4pt}
\renewcommand{\arraystretch}{0.92}
\begin{tabular}{@{}l  rrr rrr rrr@{}}
\toprule

\multicolumn{1}{c}{\multirow{2}{*}{\textbf{Category}}}
& \multicolumn{3}{c}{\textbf{MiniMax-M2.7}}
& \multicolumn{3}{c}{\textbf{GPT-5.3}}
& \multicolumn{3}{c}{\textbf{Sonnet 4.5}} \\

\cmidrule(lr){2-4} \cmidrule(lr){5-7} \cmidrule(lr){8-10}

  & $r{=}1$ & $r{=}3$ & $r{=}5$
  & $r{=}1$ & $r{=}3$ & $r{=}5$
  & $r{=}1$ & $r{=}3$ & $r{=}5$ \\
\midrule
sentiment    &  46.2 &  75.0 &  84.6  &  46.7 &  60.0 &  80.0  &  53.8 &  69.2 &  85.7 \\
analytics    &  30.8 &  58.3 &  61.5  &  26.7 &  46.7 &  60.0  &  30.8 &  53.8 &  64.3 \\
weather      &  46.2 &  76.9 &  85.7  &  40.0 &  73.3 &  73.3  &  46.2 &  73.3 &  73.3 \\
compliance   &   0.0 &   9.1 &  57.1  &   0.0 &  13.3 &  46.7  &   0.0 &   0.0 &  20.0 \\
ai\_llm      &  18.2 &  57.1 &  69.2  &  13.3 &  53.3 &  53.3  &  14.3 &  46.7 &  46.7 \\
image\_gen   &  50.0 &  64.3 &  78.6  &  33.3 &  60.0 &  80.0  &  38.5 &  64.3 &  85.7 \\
document     &  64.3 &  71.4 &  93.3  &  60.0 &  73.3 &  93.3  &  64.3 &  73.3 &  93.3 \\
dns\_network &  16.7 &  30.0 &  41.7  &  13.3 &  20.0 &  40.0  &  10.0 &  27.3 &  46.2 \\
security     &  57.1 &  72.7 &  76.9  &  60.0 &  80.0 &  80.0  &  46.7 &  66.7 &  80.0 \\
wallet       &   7.1 &  36.4 &  42.9  &   6.7 &  26.7 &  46.7  &   6.7 &  26.7 &  46.7 \\
price\_data  &   7.7 &  55.6 &  50.0  &   6.7 &  40.0 &  40.0  &   7.7 &  42.9 &  42.9 \\
crypto       &   7.7 &  23.1 &  15.4  &   6.7 &  20.0 &  20.0  &   6.7 &  13.3 &  13.3 \\
\midrule
\multicolumn{1}{c}{\textbf{Mean}} &
\textbf{29.3} & \textbf{52.5} & \textbf{63.1}
 & \textbf{26.1} & \textbf{47.2} & \textbf{59.4}
 & \textbf{27.1} & \textbf{46.5} & \textbf{58.2} \\
\bottomrule
\end{tabular}
\end{table}

Competition density reduces but does not remove the effect. Sparse categories such as sentiment and weather are highly exposed, while dense categories such as crypto are more resistant. E2 shows the Sybil effect: as the attacker adds realistic aliases, aggregate selection increases from 27.5\% at $r{=}1$ to 60.2\% at $r{=}5$, with document generation reaching 93.3\%.

We further compare E2 with a live-registry snapshot collected in April~2026.
It contains 13{,}760 endpoints across 420 domains, with the top nine domains
accounting for 87.8\% of registrations and one domain accounting for 77.5\%.
We anonymize domains because the relevant issue is concentration rather than
attribution.\footnote{The anonymized snapshot is included in our artifact:
\url{https://anonymous.4open.science/r/x402-attack-FDF1/results/attack4/cdp_bazaar_catalog_pretty.json}.}

\vspace{5pt}
{\setlength{\fboxsep}{5pt}
\noindent\fcolorbox{black!0}{orange!10}{%
\begin{minipage}{0.94\linewidth}
\textcolor{red!65!black}{Takeaways.} 
x402 can fail before payment begins. If discovery steers the agent to an adversarial shortlist, the payment flow inherits that choice. Metadata validation, ranking controls, and Sybil resistance should be treated as security requirements.
\end{minipage}}}

\subsection{Cross-Implementation Audit}
\label{sec:audit}

The audit connects the formal model to implementation practice. Instead of treating SDK issues as isolated bugs, we examine whether real deployments enforce the assumptions required by our theorems: confirmation gating, facilitator-bound settlement, pre-grant uniqueness, resource binding, settle-before-grant ordering, and cache hygiene.

\heading{Experimental setup.}
Instead of reproducing a single exploit trace, we ask whether existing implementations close the underlying gaps. We audit three x402 SDK families: the Coinbase reference TypeScript stack (SDK-A), a third-party Python integration (SDK-B), and a third-party Rust middleware implementation (SDK-C). We also probe four anonymized live testnet endpoints using the same Base Sepolia client as in the earlier experiments.

We cover six properties: confirmation gating before grant, timestamp-window enforcement, $\mathit{pay\_id}$ idempotency, $\mathit{resource\_id}$ binding, settlement completeness (verify and settle before grant), and cache-control protection on payment-gated responses. 

\heading{Results} (Table~\ref{tab:audit}).
The audit shows that basic checks are more common than boundary-closing ones. All three SDKs enforce timestamp windows, but most of the stronger cross-layer checks are missing. Confirmation gating is shallow: SDK-A and SDK-C's pessimistic mode gate at $k{=}1$, while SDK-B and SDK-C's optimistic mode grant at $k{=}0$; these defaults do not provide sufficient reorg resistance under the Attack~I-A model (\S\ref{sec:eval-attack1}). The same pattern appears in replay and binding logic: none of the audited stacks binds payments to the intended resource, and $\mathit{pay\_id}$ idempotency is absent except that SDK-C's pessimistic mode avoids duplicate grants by settling synchronously before responding.

\begin{table}[H]
\centering
\caption{SDK audit results.}
\label{tab:audit}
\small
\begin{threeparttable}
\resizebox{\linewidth}{!}{%
\begin{tabular}{rcccc}
\toprule
\multicolumn{1}{c}{\textbf{Check}} & \textbf{SDK-A (TS)} & \textbf{SDK-B (Py)}
  & \textbf{SDK-C (Rust)} & \textbf{Sev.} \\
\midrule
$k$-confirmation gating
  & \xmark\tnote{$\dagger$} & \xmark & \xmark\tnote{$\dagger$} & H \\
Timestamp window
  & \cmark & \cmark & \cmark & H \\
$\mathit{pay\_id}$ idempotency
  & \xmark & \xmark & \xmark\tnote{*} & C \\
$\mathit{resource\_id}$ binding
  & \xmark & \xmark & \xmark & C \\
Settle before grant
  & \cmark\tnote{c} & \xmark\tnote{a} & $\sim$\tnote{b} & C \\
Cache-Control
  & \xmark & \xmark & \xmark & M \\
\bottomrule
\end{tabular}%
}
\begin{tablenotes}[flushleft]
\footnotesize
\item[] \cmark~=~enforced; \xmark~=~missing; $\sim$~=~partial. Severity: C~=~Critical, H~=~High, M~=~Medium.
\item[a] SDK-B completes before settlement; later failures are dropped. $^b$ SDK-C fails only in 
\item[] optimistic mode. $^c$ SDK-A settles before grant but does not wait long enough for reorg
\item[] resistance. $^*$ SDK-C fails $\mathit{pay\_id}$ idempotency only in optimistic mode.
\item[$\dagger$] SDK-A/C in pessimistic mode gate at $k{=}1$;  SDK-B/C in optimistic mode gate at $k{=}0$.
\end{tablenotes}
\end{threeparttable}
\end{table}

The row-level results show the practical impact. SDK-B violates settle-before-grant because its streaming path finishes before \texttt{settle\_payment} runs, so later settlement failures are dropped. Under nonce pre-consumption, this produces DGR$\,{=}\,n$ without any successful on-chain settlement. SDK-C depends on its mode: pessimistic mode returns HTTP~402 when \texttt{settle()} fails, while optimistic mode sends settlement in a detached \texttt{tokio::spawn} task, allowing the client to receive 200~OK even if settlement later fails.

The cross-implementation view reveals two broader issues. First, no audited SDK binds a payment to the intended resource: a payment signed for resource~$A$ on Endpoint-1 can be replayed for resources~$B$, $C$, and $D$ on the same server. Second, deployed settlement paths do not consistently bind \texttt{msg.sender} to the payer-endorsed facilitator. In EIP-3009, any observer can submit the signed \texttt{transferWithAuthorization} payload. In Permit2, \texttt{x402ExactPermit2Proxy} omits caller restriction in \texttt{settle()}, while \texttt{x402UptoPermit2Proxy} binds the facilitator through its Witness and enforces it.

Cache hygiene is also inconsistent. No SDK automatically adds \texttt{Cache-Control: no-store} or \texttt{private} to payment-gated responses, consistent with the leakage observed when Endpoint-1 exposed publicly cacheable \texttt{Cache-Control}. Live endpoints do enforce some simpler checks: malformed headers return HTTP~402, zero or underpaid amounts are rejected, and expired timestamps are blocked by the facilitator's off-chain \texttt{validBefore} check before any on-chain call.

\vspace{5pt}
{\setlength{\fboxsep}{5pt}
\noindent\fcolorbox{black!0}{orange!10}{%
\begin{minipage}{0.94\linewidth}
\textcolor{red!65!black}{Takeaways.}
The audit shows that x402's risks stem from HTTP-chain boundary rather than isolated SDK bugs. They should be closed at the protocol/implementation levels by default.

\end{minipage}}}

\subsection{Attack costs}
\label{subsec-cost}

Our attacks do not require large upfront capital. For Attack~I-A, we do not price chain reorgs: the relevant failure is that optimistic execution exposes a pre-finality window that can be hit by existing reorg or settlement uncertainty, without any additional x402 payment. Attack~I-B only requires an observed payment authorization and a preempting L2 transaction fee, typically at the cent level (e.g., $\sim\$10^{-2}$). Attack~II requires one valid x402 payment, but the same payment can yield many HTTP grants if pre-grant idempotency is missing (e.g., one $\$0.001$ payment produced hundreds of grants in our live test). Attack~III needs at most one paid request to create a cacheable response, after which unpaid requests to the same URL can be served from cache (about one service price, e.g., $\sim\$10^{-3}$ in our test). Attack~IV has almost no direct on-chain cost because metadata manipulation and Sybil listings only require attacker-controlled endpoint entries, not settlement transactions (protocol-side cost $\approx\$0$, excluding hosting or domain costs).

\section{Mitigations}
\label{sec:mitigation}

x402 needs tighter constraints at three layers: the payment object, the server grant path, and the deployment environment. We propose six mitigations mapped to these layers (Table~\ref{tab:attack-mitigation-map}).

M1--M3 secure the payment capability by binding the resource, the authorized redeemer, and the used-payment state. M4 controls when the service may be released. M5--M6 harden the surrounding HTTP and discovery boundary.

\heading{M1: Canonical encoding and nonce/timestamp.}
Payment payloads should use a single canonical typed encoding, such as EIP-712-style structured data, and sign at least pay\_id, resource\_id, facilitator, amount, token, chain\_id, and expiry under one typed structure. Servers and facilitators should reject malformed encodings, expired timestamps, duplicate pay\_id values, and out-of-window nonces before settlement. This reduces cross-SDK serialization drift and makes both \codefield{resource\_id} scope and \codefield{facilitator} identity explicit signed properties rather than deployment-local conventions.

\begin{table}[t]
\centering
\caption{Attack-to-mitigation design map.}
\label{tab:attack-mitigation-map}
\footnotesize
\renewcommand{\arraystretch}{1}
\begin{tabular}{ccl}
\toprule
\textbf{Attack} & \textbf{Map.} & \multicolumn{1}{c}{\textbf{Rule}} \\
\midrule
\cellcolor{red!11}  All & \cellcolor{orange!13} M1 & Typed fields and freshness validation. \\
\cellcolor{red!11}  I-A & \cellcolor{orange!13}M4 & Release service only after reserve or $k$-finality. \\
\cellcolor{red!11}  I-B & \cellcolor{orange!13}M2 & Require caller to match the endorsed facilitator. \\
\cellcolor{red!11}  II & \cellcolor{orange!13}M3 & Bind resource scope and claim the payment once before grant. \\
\cellcolor{red!11}  III & \cellcolor{orange!13}M5 & Use no-store caching and isolate payment headers. \\
\cellcolor{red!11} IV & \cellcolor{orange!13} M6 & Validate metadata and limit Sybil steering. \\
\bottomrule
\end{tabular}
\end{table}

\heading{M2: Facilitator-bound settlement.}
Across settlement paths, the caller with payment capabilities must be the payer-endorsed facilitator. On Permit2 paths, \texttt{Witness} must include a \codefield{facilitator} field and \texttt{settle()} must enforce \texttt{msg.sender == witness.facilitator}; on EIP-3009 paths, settlement should be routed via a facilitator-bound wrapper or equivalent entry point that performs the same check before invoking \texttt{transferWithAuthorization}. This closes the preemption path behind Attack~I-B. Private relays may reduce mempool-level visibility, but they are only secondary hardening: they do not stop a request-path observer or Byzantine server that has already seen a usable \texttt{X-PAYMENT} header.

\heading{M3: Single-use grants and resource binding.}
Servers should treat X-PAYMENT as a single-use capability and atomically claim the \codefield{pay\_id}--\codefield{resource\_id} pair before releasing any protected response. Resource binding and idempotency are distinct checks: the server must first verify that the signed resource\_id matches the requested endpoint, then reject duplicate claims for the same pair within a bounded TTL window. This prevents cross-resource first use and replay-driven DGR$\,{=}\,n$ over-granting, including the grant-before-settle behavior observed in audited SDKs. We provide example code below to show the intended order of checks.

\medskip
{\setlength{\fboxsep}{7pt}%
\noindent\hspace*{1.4em}\fcolorbox{black!0}{black!3}{%
\parbox{0.85\columnwidth}{%
\scriptsize\ttfamily\raggedright
const payment = parseXPayment(req.header("X-PAYMENT"));\\
const intent =\\
\hspace*{1.4em}await \codefield{verifyTypedPayment}(payment, requirements);\\[0.4ex]
if (intent.resourceId !== req.path) return reply402(res);\\
if (intent.expiresAt <= Date.now()) return reply402(res);\\[0.4ex]
const claimKey =\\
\hspace*{1.4em}[intent.payId, intent.resourceId].join(":");\\
if (!(await \codefield{claims.claim}(claimKey, ttlMs)))\\
\hspace*{1.4em}return rejectReplay(res);\\[0.4ex]
const verdict =\\
\hspace*{1.4em}await facilitator.verify(payment, requirements);\\
if (!verdict.isValid) return reply402(res);\\
if (policy === "conservative") \{\\
\hspace*{1.4em}await \codefield{waitForFinality}(verdict.txHash, k);\\
\} else \{\\
\hspace*{1.4em}await \codefield{reserve}(payment, requirements);\\
\}\\[0.4ex]
res.set("Cache-Control", "no-store, private");\\
return grantPaidResponse(res);
}}}
\medskip

\heading{M4: Two-phase settlement.}
To reduce Attack~I-A risk, low-latency deployments should use a reserve-then-settle pattern, or an equivalent two-phase flow in which the server grants only after observing a concrete settlement precondition. Higher-value resources should instead wait for $k$ confirmations according to Corollary~\ref{cor:depth}. The goal is to separate ``looks payable'' from ``safe to release,'' so that operators choose the residual revert-grant exposure explicitly rather than inheriting it from facilitators.

\heading{M5: Cache and header hygiene.}
Payment-gated responses must carry \codefield{Cache-Control}: \texttt{no-store} or \texttt{private}, and paid routes should bypass intermediary caches where possible. Proxies and CDNs should also avoid normalizing, merging, or duplicating X-PAYMENT values across requests, since cache leakage and parser ambiguity can break payment-service correspondence even when chain settlement is correct.

\heading{M6: Agent-selection defenses.}
For Bazaar-style deployments, mitigation must also cover discovery. Discovery layers should validate metadata, reject unverifiable or prompt-injection-like listing claims, apply Sybil-resistant registration or reputation weighting, and diversify ranking so that no single provider can dominate the shortlist. These defenses do not replace protocol binding, but they reduce the steering surface exposed by Attack~IV. Together, M1--M6 make x402's security boundary explicit: they specify \emph{what} is being paid for, \emph{who} may settle it, and \emph{when} service may be released.

\section{Discussion}
\label{sec:discussion}

In this part, we extend our findings.

\subsection{Security--Latency Tradeoff}

Our results show a security--latency tradeoff in x402 deployment: increasing $k$ reduces measured RGP$_k$ in our local/model-based Attack~I evaluation, but increases the grant-to-finality delay, as summarized in Figure~\ref{fig:attack1-results}. In the fixed $T_b{=}2$\,s sweep, median $T_{\mathrm{gf}}$ grows from 6.0\,s at $k{=}3$ to 25.1\,s at $k{=}12$, while RGP$_k$ stays at the CI floor. In the overnight $T_b{=}12$\,s sweep, median $T_{\mathrm{gf}}$ reaches 34.3\,s, 71.5\,s, and 144.5\,s for $k{=}3,6,12$, respectively. These values should be interpreted as guidance under our controlled reorg model and timing assumptions, not as live-network frequency estimates. In practice, operators should choose $k$ based on resource value, chain characteristics, and acceptable user-visible latency (Corollary~\ref{cor:depth}).


\subsection{Threats to Validity}
\label{subsec-validity}

\heading{Measurement realism.}
Our Attack~I study mainly uses a controlled local-chain model rather than naturally occurring public-chain reorganizations. In the main sweep, revert-grant risk is estimated using the analytic model $P(\mathrm{revert}) = p_{\mathrm{reorg}}^{k}$. The Hardhat snapshot/revert injector is used for concrete trace validation, since repeated full-chain rollbacks are too costly at sweep scale. Base Sepolia provides a timing sanity check, but we observed no natural reorgs during the measurement window. Thus, the reported RGP$_k$ values should be interpreted as controlled estimates rather than live-network frequencies.

\heading{Coverage.}
Our implementation coverage is meaningful but not exhaustive. The audit covers three SDK families (TypeScript, Python, and Rust) and four live testnet endpoints, but other stacks may expose different behaviors. The Attack~IV study covers three LLMs, 12 application categories, and 2{,}160 discovery decisions over an x402scout-based catalog. This is sufficient to show a consistent effect in our setting, but not to claim that all agent architectures, ranking pipelines, or registries behave similarly.

\heading{Metric scope.}
Our metrics target the violation events defined in this paper: RGP$_k$ for revert-grants, DGR for replay-driven over-granting, cache leakage incidence for Attack~III, and selection rate for Attack~IV. They capture direct protocol and deployment failures, but not downstream effects such as reputation loss, user churn, or failures across larger multi-agent workflows. The DGR-1000 live replay test estimates settlement count from balance deltas and observed block heights; this is practical, but may miss settlements that occur outside the observation window.

\subsection{Compatibility with EIPs}
\label{subsec-withEIP}

We do not evaluate an EIP-native x402 implementation in this paper, but
emerging intent standards provide a concrete migration direction for the
protocol. EIP-8004\footnote{\url{https://eips.ethereum.org/EIPS/eip-8004}} (`Trustless Agents') points toward typed payment intents rather than deployment-local JSON payloads. For x402, the main value is not the label of the standard by itself but the opportunity to make security-critical bindings explicit and signed under a common schema. An illustrative x402-compatible typed intent could look as follows:

\medskip
\begin{center}
\fcolorbox{black!0}{black!3}{%
\parbox{0.95\columnwidth}{%
\small
\texttt{\{}\\
\hspace*{0.9em}\codefield{facilitator}\texttt{: 0xFac...,}\\
\hspace*{0.9em}\codefield{resource}\texttt{: "/v1/weather?city=Sydney",}\\
\hspace*{0.9em}\texttt{token: USDC,\ amount: 10000,\ expiry: 1712345678,}\\
\hspace*{0.9em}\codefield{hook}\texttt{: "x402.verifyAndSettle"}\\
\texttt{\}}
}}
\end{center}
\medskip

The highlighted fields show where typed intents help. A signed \codefield{facilitator} field provides a binding point for Attack~I-B, while a signed \codefield{resource} field supports the payment--resource coupling needed for Attack~II. An explicit \codefield{hook} field, used here as an x402-specific integration hook rather than a claim about EIP-8004's current field names, makes the verification and settlement entry point visible to wallets, facilitators, and resource servers under the same typed structure. Typed-intent standards, therefore, offer x402 a path toward canonical encoding, stronger cross-implementation binding, and more auditable interoperability, although they do not eliminate timing or caching risks.

At the signing layer, this is consistent with EIP-712, which defines typed structured-data signing, and with ERC-3009, whose \texttt{transferWithAuthorization} path already uses typed-signature machinery in our study. ERC-8004 is therefore best viewed not as a payment standard for x402 itself, but as a complementary agent-registry layer that could carry richer trust, discovery, and payment-related metadata around x402-enabled services.

\subsection{Design Implications for Protocol Developers}
\label{app:design}

Our findings suggest that x402 implementations should treat payment verification as a cross-layer authorization process, not as a simple HTTP header check. A server should decide how long to wait before granting access based on the value of the requested resource and the stability of the underlying chain. Low-value API responses may tolerate a small confirmation depth, while higher-value resources should wait for stronger confirmation before release.

Implementations should also make each payment usable only for the intended resource. A server should record the tuple $(\mathit{pay\_id}, \mathit{resource\_id})$ and check it before forwarding the payment to the facilitator. This prevents the same payment proof from being reused to obtain multiple grants or to access a different resource. To avoid inconsistent behavior across SDKs, the ecosystem should also converge on a single canonical payment encoding, such as EIP-712 or a fixed JSON serialization, so that all implementations verify the same signed content in the same way.

Finally, paid responses should not be cacheable by shared middleware or edge proxies. Payment-gated routes should emit \texttt{Cache-Control: no-store} or \texttt{private} by default and bypass intermediary caches where possible. Implementations should also keep structured audit logs, including $\mathit{pay\_id}$, $\mathit{resource\_id}$, $\mathit{tx\_hash}$, and timestamp, so that duplicate grants, failed settlements, and user disputes can be investigated after the fact.

\subsection{Limitations and Future Work}
\label{subsec-limits}

Our formal and empirical scope stops at the protocol/facilitator boundary: fully malicious resource servers ($\mathcal{A}_{\mathcal{R}}$) and colluding $R$/$F$ pairs are out of scope. Hardening x402 against those adversaries would likely require client-visible settlement receipts or independent on-chain verification, rather than facilitator trust alone.

For Attack~IV, the main paper validates only E1 (metadata manipulation) and E2 (Sybil flooding). The remaining marketplace threats E3--E5, such as price/latency manipulation and facilitator collusion, remain open evaluation targets. We view these as future work, together with broader ecosystem conformance testing and integration with A2A/x402 agent frameworks.

\section{Related Work}
\label{sec-rw}

We situate our work in five research strands.

\heading{Micropayment and payment channels.}
Layer-two protocols~\cite{gudgeon2020sok,qi2025sok} and payment channels~\cite{aumayr2022sleepy,avarikioti2020cerberus} provide off-chain scaling with on-chain dispute resolution. Probabilistic micropayments~\cite{almashaqbeh2020microcash,takahashi2021probabilistic} reduce settlement overhead via randomized redemption. State-channel protocols (Lightning, Raiden) require on-chain deposits and off-chain signed updates, with security depending on time-bound revocation. Voucher-based models (e.g., Filecoin retrieval markets) use sequence-numbered off-chain vouchers redeemed on-chain.

x402 differs fundamentally: it couples HTTP semantics with per-request on-chain settlement via a one-shot signed payload. No off-chain channel state is maintained, which simplifies deployment but exposes the HTTP--blockchain boundary to replay, duplication, and timing inconsistencies absent from channel-based designs.

\begin{table}[t]
\centering
\caption{Our x402 attacks vs.\ related attacks.}
\label{tab:comparison}
\small
\setlength{\tabcolsep}{3.5pt}
\renewcommand{\arraystretch}{1.}
\begin{threeparttable}
\begin{tabular}{clccccc}
\toprule
\textbf{Target} &
\multicolumn{1}{c}{\textbf{Attack}} & \textbf{\ding{172}} & \textbf{\ding{173}} & \textbf{\ding{174}} & \textbf{\ding{175}} & \textbf{\ding{176}} \\
\midrule
Payment channel & Payout race~\cite{weintraub2024payout} & -- & \cmark & -- & \cmark & F \\
Payment channel & Atomicity~\cite{wu2025atomicity} & -- & \cmark & -- & \cmark & F \\
PoS/PoA & Reorg~\cite{zhang2023time,li2025does} & -- & \cmark & -- & \pmark & E \\
Agent tools & Metadata steering~\cite{faghih2025tool,sneh2025tooltweak,mo2025ama} & -- & -- & -- & \pmark & E \\
MCP & Tool poisoning~\cite{wang2025mcptox,li2026mcpitp,zhang2025msb,lee2026overthinking} & -- & -- & -- & \pmark & E \\
\midrule
\textbf{x402 (ours)} & I: Facilitator race & -- & \cmark & -- & \cmark & E+F \\
\textbf{x402 (ours)} & II: Replay & \cmark & -- & -- & \cmark & E+F \\
\textbf{x402 (ours)} & III: Header/Proxy & -- & -- & \cmark & -- & E \\
\textbf{x402 (ours)} & IV: Server selection & -- & -- & -- & \cmark & E \\
\bottomrule
\end{tabular}
\begin{tablenotes}[flushleft]
\footnotesize
\item \cmark direct/primary coverage. \pmark partial or adjacent coverage. -- not addressed.
\item \ding{172} \textbf{Replay.} \ding{173} \textbf{Reorg/Time.} \ding{174} \textbf{Proxy/HTTP.}
\item \ding{175} \textbf{Economic impact.} \ding{176} \textbf{Evidence} (E empirical, F formal, E+F both).
\end{tablenotes}
\end{threeparttable}
\end{table}

\heading{Web payment and monetization standards.}
The W3C Payment Request API~\cite{w3c2022paymentrequest} mediates checkout via the user agent, Web Monetization~\cite{webmonetization2025} targets continuous content micropayments, and Open Payments~\cite{openpayments} defines wallet-addressed REST APIs with grant-based authorization. Unlike x402, none use HTTP~402 as a per-request handshake binding payment authorization to an application request via a one-shot signed payload, which exposes a distinct HTTP--blockchain trust boundary.

\heading{API security.}
Prior work on API monetization studies behavioral breaking changes~\cite{jayasuriya2024understanding}, agent-ready REST APIs~\cite{mastouri2025rest}, and shared storage exploitation~\cite{nisenoff2025exploiting}. x402 extends this by tying API access to on-chain payments, introducing failure modes at the HTTP--blockchain boundary not covered by traditional API security.

\heading{Cross-layer atomicity and temporal blockchain risk.}
Payment-channel work addresses synchrony assumptions~\cite{avarikioti2021brick,ersoy2022syncpcn} and cross-layer atomicity~\cite{weintraub2024payout,wu2025atomicity}, while PoS attacks~\cite{schwarz2022three,zhang2024max}, BSC finality failures~\cite{li2025does}, and reorg-resilient attestation~\cite{zhang2025available} expose the gap between \emph{included} and \emph{finalized} transactions. x402's optimistic execution creates an analogous temporal gap, which our Attack~I quantifies via RGP$_k$ and $T_{\mathrm{gf}}$.

\heading{Agent tool-selection and MCP attack surfaces.}
Agent tool choice is highly sensitive to text-visible metadata: editing descriptions shifts usage~\cite{faghih2025tool}, and ToolTweak~\cite{sneh2025tooltweak} and Attractive Metadata Attack~\cite{mo2025ama} automate such steering, with ToolCert~\cite{yeon2025quantifying} certifying robustness under deceptive injection. MCP-focused studies broaden the threat model to registry and execution-layer compromise~\cite{hou2025mcp,wang2025mcptox,li2026mcpitp,zhang2025msb,lee2026overthinking,jiang2026sok}. Our Attack~IV specializes this surface to x402's Bazaar discovery layer, focusing on economic steering of server selection rather than arbitrary tool execution.

\section{Conclusion}
\label{sec-conclusion}

x402 is not yet fully secure by construction. We identified five practical attacks leading to either \textit{unpaid-service} or \textit{paid-but-denied} without much cost. We validated attacks experimentally and empirically and propose concrete mitigations.

\bibliographystyle{unsrt}
\bibliography{bib}

\appendix

\section*{Open Science}

We host our evaluation artifacts at \textcolor{teal}{\url{https://anonymous.4open.science/r/x402-attack-FDF1}}.

\section{Full Proofs of Security Theorems}
\label{app:proofs}

We show the full proofs for the theorem statements in \S\ref{sec:theorems}.

\subsection{Proof of Theorem~\ref{thm:auth-sound}}
\label{app:proof-auth-sound}

Theorem~\ref{thm:auth-sound} presents the theorem of \textit{Authorization Soundness-Conservative Execution}.

\begin{proof}[Proof of Theorem~\ref*{thm:auth-sound}] We prove that under conservative execution, any grant issued by the resource server must correspond to a payment transaction that has already reached depth $k$ in the blockchain view of an honest facilitator. Therefore, any authorization failure after the grant necessarily implies a post-grant chain reorganization of depth at least $k$, which is bounded by the chain finality error $\epsilon_{\mathrm{chain}}(k)$. We instantiate this bound under exponential finality decay to obtain negligibility in the security parameter.

\textit{\underline{Step 1:} Grant implies $k$-deep confirmation at grant time.}
Under conservative execution, $R$ issues $R.\mathrm{grant}(C,\mathit{res},t_g)$
only when $F$ reports $\mathsf{Confirmations}(tx_{pp}, t_g) \ge k$.
Since $F$ is honest (Assumption~\ref{asm:facilitator-det}), this report faithfully reflects
$F$'s view of $\mathcal{B}$: at time $t_g$, $tx_{pp}$ has been
included in some block $b$ at depth $\ge k$ in the chain known to $F$.
Formally, let $\mathsf{Depth}(tx,t)$ denote the depth of $tx$'s
containing block in the canonical chain at time $t$. Then:
\[
R.\mathrm{grant}(C,\mathit{res},t_g)
\;\Rightarrow\;
\mathsf{Depth}(tx_{pp}, t_g) \ge k. \tag{$\dagger$}
\]

\textit{\underline{Step 2}: $E_{\mathrm{auth}}$ requires a post-grant reorganization.}
$E_{\mathrm{auth}}$ occurs when $tx_{pp}$ eventually becomes absent
from the canonical chain despite the grant at $t_g$.
By~($\dagger$), $tx_{pp}$ was at depth $\ge k$ at time $t_g$.
For $tx_{pp}$ to later become absent, the chain must undergo a
reorganization of depth $\ge k$ at some time $t' > t_g$:
\[
E_{\mathrm{auth}}
\;\subseteq\;
\bigl\{ \exists\, t' > t_g : \mathsf{Reorg}(\mathrm{depth} \ge k,\, t') \bigr\}.
\]

\textit{\underline{Step 3}: Bounding the reorganization probability.}
By Assumption~\ref{asm:blockchain-indep}, the event $\mathsf{Reorg}(\mathrm{depth} \ge k, t')$
is independent of the grant event conditional on the broadcast of
$tx_{pp}$, and in particular is not influenced by $F$'s reporting
behavior or HTTP-layer scheduling.
Therefore:
\begin{align*}
\Pr[E_{\mathrm{auth}}]
&\le \Pr\!\bigl[\exists\, t' > t_g :
     \mathsf{Reorg}(\mathrm{depth} \ge k,\, t')\bigr] \\
&\le \Pr[\mathsf{Reorg}(\mathrm{depth} \ge k)]
\;\le\; \epsilon_{\mathrm{chain}}(k),
\end{align*}
where the last inequality is Definition~\ref{def:chain-finality}.

\textit{\underline{Step 4}: Parameter choice for negligibility.}
For exponential decay $\epsilon_{\mathrm{chain}}(k) = e^{-\alpha k}$,
set $k^* = \lceil\alpha^{-1}(\lambda+1)\ln 2\rceil$.
Then $e^{-\alpha k^*} \le e^{-(\lambda+1)\ln 2} = 2^{-(\lambda+1)} < 2^{-\lambda}$,
which is negligible in $\lambda$.

\end{proof}

\subsection{Proof of Theorem~\ref{thm:auth-fail}}
\label{app:proof-auth-fail}

Theorem~\ref{thm:auth-fail} refers to the theorem of \textit{Authorization Soundness Violation-Optimistic Execution}.

\begin{proof}[Proof of Theorem~\ref*{thm:auth-fail}]
Let $t_0$ be the time $C$ broadcasts $tx_{pp}$ to the mempool.
Define the events
\begin{align*}
    \mathsf{VulnWindow} 
 = \bigl\{t_g < t_0 + T_{\mathrm{inc}} + k \cdot T_b\bigr\}, &\\
  \mathsf{Reorg}_k
  = \bigl\{\mathcal{A}_{\mathcal{B}} \text{ causes a reorg during }
             [t_g,\,
             t_0 + T_{\mathrm{inc}} + k \cdot T_b]\bigr\}. &
\end{align*}

\textit{\underline{Step 1}: Characterize the grant time.}
Under optimistic execution, the resource server grants at
$t_g = t_0 + T_{\mathrm{verify}} + \delta$.
The grant precedes $k$-confirmation finality exactly when
\[
t_0 + T_{\mathrm{verify}} + \delta
< t_0 + T_{\mathrm{inc}} + k \cdot T_b,
\]
which is equivalent to
$T_{\mathrm{inc}} + k \cdot T_b > T_{\mathrm{verify}} + \delta$.

\textit{\underline{Step 2}: Revert-grant requires both a vulnerable window and a reorganization.}
If a grant occurs before finality and the transaction is later removed
from the canonical chain, then the execution contains both
$\mathsf{VulnWindow}$ and $\mathsf{Reorg}_k$. Hence
$E_{\mathrm{revert}} \supseteq \mathsf{VulnWindow} \cap \mathsf{Reorg}_k$.

\textit{\underline{Step 3}: Apply independence.}
By Assumption~\ref{asm:blockchain-indep}, the chain reorganization
process is independent of HTTP timing once the set of broadcast
transactions is fixed. Therefore
\begin{align*}
\mathrm{RGP}_k
&= \Pr[E_{\mathrm{revert}}] \\
&\ge \Pr[\mathsf{VulnWindow} \cap \mathsf{Reorg}_k] \\
&= \Pr[\mathsf{Reorg}_k] \cdot \Pr[\mathsf{VulnWindow}] \\
&\ge p_{\mathrm{reorg}} \cdot
   \Pr\!\bigl[T_{\mathrm{inc}} + k \cdot T_b
              > T_{\mathrm{verify}} + \delta\bigr].
\end{align*}

\textit{\underline{Step 4}: Special case $k=0$.}
Setting $k=0$ yields
\[
\mathrm{RGP}_0
\ge p_{\mathrm{reorg}} \cdot
\Pr[T_{\mathrm{inc}} > T_{\mathrm{verify}} + \delta].
\]
\end{proof}

\subsection{Proof of Theorem~\ref{thm:replay}}
\label{app:proof-replay}

Theorem~\ref{thm:replay} refers to the \textit{Replay Resistance under Pre-Grant Uniqueness Enforcement}.

\begin{proof}[Proof of Theorem~\ref*{thm:replay}] We proceed with two directions. 

\underline{\textit{Positive direction.}}
Let $t$ be the time $R$ first accepts $\mathit{PP}$ with identifier
$\mathit{pay\_id}$. The server records
$\mathcal{I}[\mathit{pay\_id}] \leftarrow (\mathit{resource\_id}, t)$.

\emph{Case 1: Replay within the TTL window.}
For any replay arriving at time $t' \in (t, t + \mathrm{TTL}]$, the
lookup in $\mathcal{I}$ occurs before any further verification or grant.
Because the entry has not expired, the lookup succeeds and the request
is rejected before access is granted. Hence $E_{\mathrm{replay}}$ cannot
occur in this interval.

\emph{Case 2: Replay after the TTL window.}
Once the $\mathcal{I}$ entry expires, the timestamp defense applies.
The embedded timestamp satisfies $\mathit{ts} \approx t$, while the
server rejects any payload with
$\mathit{ts} \notin [t' - W, t' + W]$.
Since $t' > t + \mathrm{TTL}$ and $W < \mathrm{TTL}$, we have
$t < t' - W$, so $\mathit{ts} < t' - W$, causing rejection by the
freshness check. Thus the theorem's stated guarantee on
$t' \in (t, t + \mathrm{TTL}]$ holds, and the combined defenses are in
fact stronger.

\underline{\textit{Negative direction.}}
Without $\mathcal{I}$, each replayed request is processed independently
by $R$. The resource server forwards
$\langle \mathit{PP}, \mathit{PR} \rangle$ to $F$; the facilitator
verifies the still-valid signature and returns an approval, so each
request triggers a fresh grant. If this is repeated $n$ times, then
\[
\mathsf{GrantCount}(\mathit{pay\_id}, \mathcal{E}) = n
\]
with probability~$1$.

\end{proof}


\subsection{Proof of Theorem~\ref{thm:atomicity}}
\label{app:proof-atomicity}

Theorem~\ref{thm:atomicity} presents the \textit{Facilitator $k$-Atomicity as Prerequisite}.

\begin{proof}[Proof of Theorem~\ref*{thm:atomicity}] 
We prove that facilitator $k$-atomicity is necessary for authorization soundness under conservative execution. If $F$ may report finality before $tx_{pp}$ actually reaches depth $k$, then the server can grant access based on an under-confirmed transaction. A later reorganization can therefore invalidate the payment after the grant, yielding a non-negligible authorization failure probability.

\textit{\underline{Step 1}: Non-atomicity induces effective under-confirmation.}
When $F$ is not $k$-atomic, there exist executions in which
$F$ reports $\mathsf{Final}(tx_{pp}, k)$ at time $t_g$
while $\mathsf{Confirmations}(tx_{pp}, t_g) = k' < k$.
Under conservative execution, $R$ grants access upon this report,
so the effective confirmation depth at grant time is $k'$, not $k$.

\textit{\underline{Step 2}: A subsequent reorganization can invalidate the grant.}
For $tx_{pp}$ at depth $k' < k$, a reorganization of depth
$\ge k'+1 \le k$ suffices to remove $tx_{pp}$ from the canonical chain.
By Assumption~\ref{asm:blockchain-indep}, this reorganization event
depends only on the chain process, not on $F$'s reporting behavior.

\textit{\underline{Step 3}: Independence argument.}
Let $\mathcal{G}$ be the $\sigma$-algebra generated by HTTP-layer
events and facilitator reports up to time $t_g$.
The event $E_{\mathrm{atom}}$ is $\mathcal{G}$-measurable, while the
post-grant reorganization event is determined by the mining process
after $t_g$. Assumption~\ref{asm:blockchain-indep} makes these events
independent.

\textit{\underline{Step 4}: Combine the two events.}
Authorization soundness fails whenever $F$ prematurely reports finality
and the chain later reorganizes deeply enough to remove $tx_{pp}$.
Therefore
\begin{align*}
\Pr[E_{\mathrm{auth}}]
&\ge \Pr[E_{\mathrm{atom}}
     \,\wedge\, \mathsf{Reorg}(\mathrm{depth} \ge k'+1)] \\
&= \Pr[E_{\mathrm{atom}}] \cdot
   \Pr[\mathsf{Reorg}(\mathrm{depth} \ge k'+1)] \\
&\ge \mu \cdot p_{\mathrm{reorg}}.
\end{align*}
Since $\mu$ is non-negligible and $p_{\mathrm{reorg}} > 0$, the
product is non-negligible violating authorization soundness.

\end{proof}


\subsection{Proof of Corollary~\ref{cor:depth}}
\label{app:proof-depth}

Corollary~\ref{cor:depth} shows the \textit{Confirmation Depth Recommendation}.

\begin{proof}[Proof of Corollary~\ref*{cor:depth}]
The first statement is immediate from
Theorem~\ref{thm:auth-sound}: under conservative execution with an
honest facilitator, any choice of $k$ guarantees
\[
\Pr[E_{\mathrm{auth}}] \le \epsilon_{\mathrm{chain}}(k).
\]
To meet a target failure probability
$\epsilon_{\mathrm{target}}$, it therefore suffices to choose the
smallest confirmation threshold satisfying
$\epsilon_{\mathrm{chain}}(k) \le \epsilon_{\mathrm{target}}$, which is
exactly the definition of $k^*$.

For the Base-style exponential calibration
$\epsilon_{\mathrm{chain}}(k) \approx e^{-\alpha k}$, this becomes
\[
k^* = \min\{k : e^{-\alpha k} \le \epsilon_{\mathrm{target}}\}
    = \left\lceil \alpha^{-1}\ln(\epsilon_{\mathrm{target}}^{-1}) \right\rceil.
\]
Instantiating this expression at the two target risk levels used in the
paper, $\epsilon_{\mathrm{target}} = 10^{-2}$ and
$\epsilon_{\mathrm{target}} = 10^{-4}$, yields the illustrative
recommendations $k \ge 3$ and $k \ge 12$ under the chain-specific Base
calibration assumed in the corollary statement.

\end{proof}

\end{document}